\documentclass[letterpaper, 10 pt, conference]{ieeeconf}  % Comment this line out if you need a4paper

\IEEEoverridecommandlockouts                              % This command is only needed if 
\usepackage{graphics} % for pdf, bitmapped graphics files
\usepackage{epsfig} % for postscript graphics files
\usepackage{mathptmx} % assumes new font selection scheme installed
\usepackage{times} % assumes new font selection scheme installed
\usepackage{amsmath} % assumes amsmath package installed
\usepackage{amssymb}  % assumes amsmath package installed
\usepackage{url}

\usepackage{algorithm}
\usepackage[noend]{algpseudocode} % from algorithmicx
\usepackage{xcolor}
\usepackage{booktabs}
\usepackage{cleveref}
\usepackage{listings}
\usepackage{cite}
\usepackage{xspace}
\usepackage{tcolorbox}
\tcbuselibrary{listings, breakable, skins}

\newtcblisting{codebox}[2][]{%
  listing engine=listings,
  listing only,
  breakable,
  enhanced,
  colback=black!1,
  colframe=black!20,
  boxrule=0.5pt,
  arc=2pt,
  left=6pt,right=6pt,top=0pt,bottom=0pt,
  listing options={
    language=#2,
    basicstyle=\ttfamily\footnotesize,
    numbers=left,
    numberstyle=\tiny,
    numbersep=5pt,
    breaklines=true,
    breakatwhitespace=false,
    columns=fullflexible,
    keepspaces=true,
    showstringspaces=false,
    tabsize=2
  },
  #1
}

\definecolor{mygray}{gray}{0.96}
\definecolor{darkgray}{gray}{0.30}

\newtcolorbox{theorybox}[1]{
  breakable,
  enhanced,
  colback=mygray,
  colframe=darkgray,
  coltitle=white,
  colbacktitle=darkgray,
  boxrule=0.9pt,
  sharp corners,
  title={#1},
  fonttitle=\bfseries,
  left=4pt,
  right=4pt,
  top=4pt,
  bottom=4pt,
  boxsep=2pt,
  toptitle=2pt,
  bottomtitle=2pt,
  before upper={
    \setlength{\abovedisplayskip}{4pt}
    \setlength{\belowdisplayskip}{4pt}
    \setlength{\abovedisplayshortskip}{2pt}
    \setlength{\belowdisplayshortskip}{2pt}
  }
}
\newcommand{\abCROWN}{$\alpha,\!\beta$-CROWN\xspace}

\newtheorem{theorem}{Theorem}

\newtheorem{definition}{Definition}

\newtheorem{example}{Example}

\title{\LARGE \bf
Bridging Control with Neural Network Verifier \texttt{alpha-beta-CROWN}: A Tutorial
}

\author{Haoyu Li$^*$, Xiangru Zhong$^*$, Hao Cheng, Bin Hu, Huan Zhang% <-this % stops a space
\thanks{* Equal Contribution}% <-this % stops a space
\thanks{Haoyu Li is with the Department of Computer Science, and Xiangru Zhong, Hao Cheng, Bin Hu, Huan Zhang are with the Department of Electrical and Computer Engineering, University of Illinois Urbana-Champaign, Urbana, IL 61801, USA
{\tt\small haoyuli5@illinois.edu, xiangru4@illinois.edu, haoc539@illinois.edu, binhu7@illinois.edu, huan@huan-zhang.com}}%
}

\begin{document}

\maketitle
\thispagestyle{empty}
\pagestyle{empty}

%%%%%%%%%%%%%%%%%%%%%%%%%%%%%%%%%%%%%%%%%%%%%%%%%%%%%%%%%%%%%%%%%%%%%%%%%%%%%%%%
\begin{abstract}

Learning-based methods for synthesizing controllers have gained popularity due to their high expressiveness and strong empirical performance. However, in safety-critical scenarios such as autonomous driving, robotics, and power systems, empirical performance alone is insufficient, and formal verification of controller properties such as stability and safety is highly desirable. Unfortunately, many prior verification approaches are either tied to specific structural assumptions on the system or the certificate, making them difficult to transfer across settings, or suffer from poor scalability on higher-dimensional neural network systems. In this tutorial, we present a unified framework that aims to mitigate this gap via bridging control with the state-of-the-art neural network verifier \abCROWN (alpha-beta-CROWN). At its core, \abCROWN is a general-purpose bounding engine for nonlinear functions represented as computation graphs: given an input domain, it can automatically produce certified bounds and explicit linear relaxation of the nonlinear function. These certified bounds are useful on their own for tasks such as reachability analysis and local linear approximation, and they also provide the foundation for more complex routines that perform verification, satisfiability checking, and optimization. More specifically, many control problems reduce to verifying real-valued inequalities over a state domain (e.g., Lyapunov theory and barrier functions). Consequently, \abCROWN enables scalable verification of such conditions by computing tight bounds on these conditions and recursively partitioning and pruning subdomains based on the bounds. Thanks to GPU parallelization, this pipeline demonstrates superior scalability on verification and optimization problems that are challenging for traditional approaches. In this tutorial, we discuss the basics of \abCROWN, and introduce its application to various control-related tasks. Finally, we will discuss training frameworks for co-synthesizing controllers and corresponding verifier-friendly certificates.

% More specifically, it enables scalable satisfiability solvers by proving real constraints in the box, and it can be extended to an optimizer by recursively partitioning and pruning out domains using the bound information. 

% This verifier checks generic properties by efficiently constructing linear bounds on the objective over the domain. To tighten these bounds, it further employs branch-and-bound to split the domain into many smaller subdomains. Because linear bounds become tighter over smaller domains, the combination of branch-and-bound with an efficient, parallelizable bound-propagation algorithm enables scalable verification. Moreover, we introduce how to convert generic control verification problems into a form that $\alpha,\!\beta$-CROWN can handle. Since $\alpha,\!\beta$-CROWN supports general computation graphs, such modifications are typically easy to implement and are not restricted to a specific problem class. Finally, we discuss training frameworks for co-synthesizing controllers and corresponding certificates that are verifier-friendly.

\end{abstract}

%%%%%%%%%%%%%%%%%%%%%%%%%%%%%%%%%%%%%%%%%%%%%%%%%%%%%%%%%%%%%%%%%%%%%%%%%%%%%%%%
\section{Introduction}
% \textcolor{blue}{say we provide dReal compatible API}
Learning-based neural network (NN) control policies have become increasingly prevalent in modern control applications, including autonomous vehicles, robotic manipulations, and power system regulations~\cite{kaufmann2023champion,zhao2024applications,glavic2019deep}. Their expressive capacity enables the synthesis of control policies for general nonlinear plants \cite{bucsoniu2018reinforcement}. However, the deployment of such controllers in safety-critical settings requires guarantees that cannot be obtained through empirical testing alone. As a matter of fact, there is a growing demand for formal verification frameworks capable of certifying closed-loop properties such as stability and safety for learned NN controllers~\cite{tran2020verification,dawson2023safe}.

% \huan{Briefly discuss what is the nature of this verification problem? before mentioning SMT}

% Since closed-loop verification conditions can often be expressed as conjunctions/disjunctions of local nonlinear constraints, a natural approach is to cast verification as a nonlinear satisfiability problem and solve it with a general-purpose SMT verifier or nonlinear solvers.

A natural formulation is to cast verification of NN controllers as a nonlinear satisfiability problem and solve it with a general-purpose SMT verifier or other nonlinear programming solvers. In this vein, existing works formulate the closed-loop verification objective as a constraint system and invoke SMT solvers~\cite{chang2019neural,jagtap2020formal,liu2025formally,liu2025physics,meng2025learning,liu2024tool,wang2024actor,edwards2024fossil,liu2023towards} (e.g., dReal~\cite{gao2013dreal} and Z3~\cite{de2008z3}), MILP/MIP-based encodings for piecewise-linear (ReLU/LeakyReLU activations) networks
~\cite{tjeng2017evaluating,katz2017reluplex,katz2019marabou,karg2020stability,dai2021lyapunov,wu2023neural}, or convex optimization problems such as semidefinite programming (SDP) or sum-of-squares (SOS) programming~\cite{papachristodoulou2005tutorial,yin2021stability,dai2024verification,newton2022stability,detailleur2025improved,yin2021imitation,clark2024semialgebraic,wei2021control,tedrake2010lqr,zhao2023safety}. However, these methods typically scale poorly due to the limited ability to exploit GPU-acceleration and parallelism. In addition, methods such as SOS/MILP can only handle problems with specific structures, for example, relatively shallow piecewise-linear networks with ReLU-type activations, or polynomials~\cite{tjeng2017evaluating,katz2017reluplex,katz2019marabou}. Recent control-oriented verification methods also exploit certificate-specific structure, for example, through symbolic bound propagation~\cite{pmlr-v270-hu25a,zhang2023exact,vertovec2025scalable,hu2024real,li2025verifiable}. While effective for their target properties, these pipelines are less directly transferable across different certificate classes. 

In the machine learning community, NN verification has emerged as an important direction for scalable certification~\cite{wong2018provable,raghunathan2018certified,brown2022unified,gehr2018ai2,singh2018fast,zhang2018efficient,jafarpour2024efficient,jafarpour2023interval}. Rather than solving the nonlinear constraints exactly as done in SMT or MILP-based approaches, one can construct sound over-approximations of the network's behavior over a given input set~\cite{gowal2018effectiveness, zhang2018efficient}. Built upon this idea, $\alpha,\!\beta$-CROWN has become a state-of-the-art NN verifier that combines linear bound propagation with branch-and-bound~\cite{zhang2018efficient,salman2019convex,xu2020automatic,xu2021fast,wang2021beta,zhang2022general,kotha2023provably,zhou2024scalable,shi2024genbab,zhou2025clip} to achieve several key advantages. First, $\alpha,\!\beta$-CROWN directly works on general nonlinear functions represented by computation graphs, enabling the direct handling of constraints and objectives with complex NN components~\cite{yang2024lyapunov,shi2024certified,li2025neural,serry2025safe,li2025two}. Secondly, $\alpha,\!\beta$-CROWN exhibits much stronger empirical scalability. Notice that $\alpha,\!\beta$-CROWN propagates sound and efficient linear over-approximations of the computation graph in a very efficient manner, and uses branch-and-bound to iteratively partition and tighten those subdomains whose bounds remain inconclusive. The efficient bound propagation in $\alpha,\!\beta$-CROWN is accelerated by GPUs, and a large number of subdomains can be calculated in parallel, which makes the verification process very efficient. Finally, the linear bound propagation algorithm in $\alpha,\!\beta$-CROWN is fully differentiable, enabling more downstream applications such as training verification-friendly controllers~\cite{shi2024certified}. 

The properties of $\alpha,\!\beta$-CROWN enable several use modes for control problems. First, the verifier could return certified numeric bounds as well as local linear bounds of a complicated nonlinear function represented as a computation graph. These bounds can be directly used for reachability analysis~\cite{chen2013flow,althoff2021set,tran2020nnv,zhang2023reachability,zhang2024reachability,abate2024safe,shen2026diffreach} and for constructing local linear approximations of nonlinear functions~\cite{coogan2015efficient,coogan2017finite,coogan2020mixed,dutreix2020specification,sidrane2022overt}. Second, the same bounds can be used for solving general satisfiability problems over real arithmetic. This enables scalable certification of many control theoretic properties, such as Lyapunov stability~\cite{meng2024zubov,meng2025towards,liu2025formal}, Barrier function safety~\cite{ames2016control,robey2020learning,dai2023convex,dawson2022safe}, or contraction~\cite{tsukamoto2020neural,tsukamoto2020stochastic,tsukamoto2021theoretical,sun2021learning}. Finally, the verifier can be used as an optimizer for general nonlinear objectives, such as MPC style optimization with general nonlinear models~\cite{shen2024bab}. Building on these capabilities, this tutorial shows how to reformulate common control-theoretic problems amenable to $\alpha,\!\beta$-CROWN.  Additionally, we illustrate how to jointly synthesize verifiable controllers together with neural certificates with counterexample-guided Synthesis (CEGIS)~\cite{abate2018counterexample,yang2024lyapunov,chang2019neural,ravanbakhsh2015counterexample,edwards2024fossil,abate2020formal,masti2023counter,dai2020counter,ravanbakhsh2016robust} or certified training~\cite{wong2018provable,mirman2018differentiable,shi2024certified}. In summary, the goals of this tutorial are:
\begin{itemize}
    \item Provide the mathematical background of the neural network verifier $\alpha,\!\beta$-CROWN.
    
    \item Demonstrate how to formulate control problems, such as reachability, satisfiability  (Lyapunov and Barrier function verification), and optimization (MPC), in an $\alpha,\!\beta$-CROWN-friendly manner.
    \item Detail the programming interface and practical usage of \abCROWN.
    % \item Introduce the API and usage of $\alpha,\!\beta$-CROWN. Demonstrate how to rewrite the control problems into an $\alpha,\!\beta$-CROWN-friendly manner. We also provide a dReal-compatible API, enabling existing dReal-based workflows to be translated seamlessly to the more scalable $\alpha,\!\beta$-CROWN framework.
    \item Introduce several training paradigms that can yield verifiable neural certificates. 
\end{itemize}

\section{Motivating Examples}\label{sec:hook}

To make the capabilities of $\alpha,\!\beta$-CROWN concrete, in this section, we illustrate each mode of $\alpha,\!\beta$-CROWN with simple motivating examples. More detailed applications to control will be deferred to Section~\ref{sec:abcrown_usage} and \ref{sec:control_app}. All the examples provided in this section will consider a discrete-time system
\begin{equation}\label{eq:discrete_time_sys}
    x_{t+1} = f(x_t, u_t)
\end{equation}
where $x_t\in\mathbb{R}^n$ is the state, $u_t\in\mathbb{R}^m$ is the control input, and $f$ may be a general nonlinear dynamics model (including learned neural dynamics). When a state-feedback controller $u_t=\pi(x_t)$ is applied, the resulting closed-loop system is
\begin{align}\label{eq:closed_loop_sys}
    x_{t+1} = f\!\bigl(x_t,\pi(x_t)\bigr) =: g(x_t),
\end{align}
where $g$ denotes the induced closed-loop map.\footnote{The complete code for all examples in this tutorial is available at \url{https://github.com/Verified-Intelligence/abCROWN_Control_Tutorial}.}
% \xiangru{Do we also need to put the link to \abCROWN somewhere in the paper?}
% \huan{Mention that full code is provided at [link in footnote]}

% In general, $\alpha,\!$

\subsection{$\alpha,\!\beta$-CROWN for Computing Certified Bounds}

% \textcolor{blue}{Be clear, we compute the bound, not verify. Give intuition.}
We first illustrate a simple use case of the bounds $\alpha,\!\beta$-CROWN produced. Given an input domain $\mathcal{B}$ for which the initial state belongs, it is oftentimes of interest to compute an overapproximation of the one-step reachable set, i.e., the set of all possible $x_{1} = g(x_0)$ for $x_0 \in \mathcal{B}$. Even this one-step problem is nontrivial when $g$ contains NN components. In this scenario, $\alpha,\!\beta$-CROWN serves to compute certified bounds of the function $g$ over $\mathcal{B}$, yielding a sound enclosure: 
\begin{theorybox}{Compute Bounds in $\alpha,\!\beta$-CROWN}
    \begin{equation}
    \text{Find } \underline{x}, \overline{x} \in \mathbb{R}^n \text{   such that   } g(\mathcal{B}) \subseteq [\underline{x}, \overline{x}],
    \end{equation}
\end{theorybox}
\noindent
where $[\underline{x}, \overline{x}] = \{x\in \mathbb{R}^n: \underline{x} \leq x \leq \bar{x}\}$ denotes an elementwise-box. Unrolling this process step by step allows scalable reachability analysis over very general function classes, including learning-based dynamics with neural networks.
\begin{codebox}{Python}
from abcrown import (
    ABCrownSolver,
    ConfigBuilder,
    IOConstraints,
    input_vars,
    output_vars,
)
# Note: These imports are needed for all the following examples.

# 1. Define and instantiate the computation graph
class ReachabilityGraph(nn.Module):
    def __init__(self):
        super().__init__()
        self.dt = 0.1
        # A residual neural network dynamics model: x_{t+1} = x_t + dt * NN(x_t)
        self.net = nn.Sequential(
            nn.Linear(1, 16),
            nn.ReLU(),
            nn.Linear(16, 16),
            nn.ReLU(),
            nn.Linear(16, 1)
        )

    def forward(self, x_t):
        # The computation graph natively supports standard PyTorch arithmetic
        return x_t + self.dt * self.net(x_t)

computation_graph = ReachabilityGraph()

# 2. Define graph input and output variables
x = input_vars(1)
y = output_vars(1)

# 3. Create solver instance with the computation graph
cfg = ConfigBuilder.from_defaults().set(
    "bab/timeout", 30)
solver = ABCrownSolver(computation_graph, x, y, config=cfg)

# 4. Create input constraints
# Bounding the initial state within a domain D
input_constraints = (x >= -1.0) & (x <= 1.0)
constraints = IOConstraints(
    input_vars=x,
    input_constraints=input_constraints
)

# 5. Compute bounds on y given constraints on x
result = solver.compute_bounds(
    constraints=constraints,
    objective=y
)

print("Lower bounds: ", result.lower)
print("Upper bounds: ", result.upper)
\end{codebox}
% \xiangru{Make the computation graph slightly more complicate (not just a ReLU network)}
% \huan{Here, explain what is obtained from abCROWN? What things are returned? Just reading the code is a bit unclear on what was going on.}

% Even when the plant dynamics are known and simple, the controller $\pi$ is often a neural network, and the certificate or specification we care about (e.g., a Lyapunov or barrier function) may also be represented by a neural network. This turns the desired guarantee into a very nonlinear constraint over a continuous region, and it must hold for all (infinitely many) states in the domain. 

\subsection{$\alpha,\!\beta$-CROWN for Solving Satisfiability Problems}

We then consider an example of using $\alpha,\!\beta$-CROWN for solving a satisfiability problem. Formally certifying properties of neural network-controlled dynamical systems, such as stability or safety guarantees, is non-trivial. For example, the standard strategy to prove the stability of~\eqref{eq:discrete_time_sys} is through Lyapunov functions~\cite{khalil2002nonlinear,lyapunov1992general}. Formally, suppose there is a domain $\mathcal{B}$ around the equilibrium $x^*$ such that there exists a function $V: \mathcal{B} \to \mathbb{R}$ such that 
\begin{enumerate}
    \item $V(x^*) = 0$ and $V(x) > 0$ otherwise,
    \item $V(g(x)) < V(x)$ for all $x \neq x^*$,
\end{enumerate}
then $x^*$ is asymptotically stable. Suppose that we have a candidate $V$ that satisfies $V(x^*) = 0$ and $V(x) > 0$ otherwise (e.g., by construction~\cite{yang2024lyapunov}), $\alpha,\!\beta$-CROWN enables efficient verification of the following satisfiability problem:
\begin{theorybox}{Verifying Satisfiability Problems in $\alpha,\!\beta$-CROWN}
\begin{align*}
    V(f(x, \pi(x))) - V(x) < \epsilon \quad \forall x \in \mathcal{B}.
\end{align*}
\end{theorybox}
\noindent
where $\epsilon > 0$ is a small tolerance used to handle the equilibrium. Here, the condition is simplified. In Section~\ref{sec:discrete_Lyapunov}, we provide a more detailed discussion on ROA verification in \abCROWN. As a high-level intuition, $\alpha,\!\beta$-CROWN converts the condition into a bound computation problem and compares the resulting bounds with $0$. Users define the verification specification using a Python-embedded DSL (Domain-Specific Language), and call $\alpha,\!\beta$-CROWN to verify it.

\begin{codebox}{python}
# 1. Define a control system and a Lyapunov function
dynamics = get_dynamics_model()    # f(x, u)
controller = get_controller()      # pi(x)
lyapunov = get_lyapunov_network()  # V(x)
    
# 2. Build the computation graph
# We assume this graph outputs a vector y:
# y[0] = V(f(x, pi(x))) - V(x) + kappa * V(x)
model = DiscreteLyapunovGraph(
    dynamics, controller, lyapunov, kappa=0.1)

# 3. Define graph input and output variables
x = input_vars(2)       # Input dimension is 2
y = output_vars(1)      # Output dimension is 1

# 4. Create solver instance with the graph
solver = ABCrownSolver(model, x, y)

# 5. Create verification constraints.
input_constraints = (
    (x >= [-2.0, -2.0]) &
    (x <= [2.0, 2.0])
)
output_constraints = (
    (y[0] < 1e-6) 
)

constraints = IOConstraints(
    input_vars=x,
    output_vars=y,
    input_constraints=input_constraints,
    output_constraints=output_constraints,
)

# 6. Run verification.
result = solver.verify(constraints=constraints)

print("Status:", result.status)
# "verified", "falsified", "unknown".
\end{codebox}
% If the result status is ``verified'', the Lyapunov condition is verified over the \emph{entire} input region. If it's ``falsified'', a counterexample is expected.
% \huan{What does abCROWN return? Need to say something like: if verify returns true then the Lyapunov condition is verified over the \emph{entire} input region, etc}
Although simple, this problem is a good general template for satisfiability problems: many control specifications reduce to verifying the sign of a scalar expression over a domain $\mathcal{B}$. In $\alpha,\!\beta$-CROWN, switching from Lyapunov to other properties usually amounts to changing only a few lines of PyTorch code, while the whole verification pipeline remains the same. While here we used a discrete-time system as an example, $\alpha,\!\beta$-CROWN can also handle continuous-time systems with Jacobians. 
% \huan{Like a hello-world example where we can introduce the idea of verification}
% \huan{Here we can use a discrete-time system. But mention continous time with Jacobian can also work, see Section XXX.}

\subsection{$\alpha,\!\beta$-CROWN for Non-linear Optimization}

Finally, we illustrate the use of $\alpha,\!\beta$-CROWN as an optimizer with an example of a finite horizon MPC problem. A common goal for planning problems in control is to find a sequence of optimal actions $u_t$ such that the sum of the step cost is minimized, i.e.
\begin{theorybox}{MPC optimization in $\alpha,\!\beta$-CROWN}
\begin{align}
    \min_{u_t \in \mathcal{U}} \sum_{t=t_0}^{t_0 + H} c(x_t, u_t), \quad \text{s.t. } x_{t+1} = f(x_{t}, u_{t})
\end{align}
\end{theorybox}
\noindent
with $c$ being the cost function. Since in many cases $f$ will be a learning-based dynamics, the optimization problem is highly non-convex and nonlinear, especially when the horizon is long. The key observation is that, although complicated, the total cost function $\sum_{t=t_0}^{t_0 + H} c(x_t, u_t)$ can be treated as a computation graph that maps $\mathcal{U}^H$ to $\mathbb{R}$. Therefore, $\alpha,\!\beta$-CROWN can efficiently and systematically handle this optimization problem. Compared to directly applying gradient descent, which is fast but typically provides only a local minimum, $\alpha,\!\beta$-CROWN is theoretically a global optimizer guided by certified lower bounds, which can be used as a signal for optimality gap. More details of the optimizer functionality are discussed in Section~\ref{sec:optimization}.

% \huan{People will say why not just optimize the objective using gradient decent? The good thing with abCROWN is that it can provide a lower bound on objective and an optimality guarantee, and is theoretically a global optimizer rather than local search.}

\begin{codebox}{Python}
# 1. Load the model
model = get_MPCGraph(horizon=H)

# 2. Define graph input and output variables
u = input_vars(H)
y = output_vars(1)

# 3. Build the solver
solver = ABCrownSolver(model, u, y)

# 4. Create constraints
input_constraints = (u >= -5.0) & (u <= 5.0)
constraints = IOConstraints(
    input_vars=u,
    input_constraints=input_constraints,
)

# 5. Minimize the total cost
result = solver.minimize(
    objective=y[0],
    constraints=constraints)

print(f"Minimum cost found: {
    result.primal_value}")
print(f"Optimal action sequence: {
    result.x_best}")
\end{codebox}

\section{The $\alpha,\!\beta$-CROWN Verifier}
Formally, \abCROWN aims to verify 
\begin{align}\label{verify:condition}
    F(x) > 0 \text{ for all } x\in \mathcal{B}
\end{align}
with $F$ composed by supported nonlinear functions and $\mathcal{B}$ being a hyper-rectangle with the same dimension as $x$. $\alpha,\!\beta$-CROWN couples fast \emph{symbolic linear bound propagation} with
\emph{branch-and-bound}. For a hyper-rectangle $\mathcal{B}$, it constructs an affine lower bound of $F$ and provides certification based on the minimum value of this affine bound. If inconclusive, it \emph{splits} $\mathcal{B}$ into sub-domains and repeats. As the domains shrink, local linear relaxations tighten, and the overapproximation error from the nonlinearities decreases. The full process is highly efficient, as all necessary computations, including local relaxations and bound propagations, are heavily vectorized and can be executed in parallel. Tens of thousands of sub-domains can be processed simultaneously on a GPU, enabling rapid pruning of large regions and certification with only coarse bounds on tiny domains.

% \medskip
% \noindent Our tutorial will go through the following aspects of $\alpha,\!\beta$-CROWN: 
% \begin{enumerate}
%     \item What functions $F$ are supported for verification?
%     \item How is the bound computed on each subdomain?
%     \item How is the split (branch-and-bound) being done?
%     \item Further improvement to tighten the bounds.
%     \item How to handle more complex logical formula for verification?
%     \item How to compute bounds for the Jacobian of a function?
%     \item Overall pipeline of $\alpha,\!\beta$-CROWN.
% \end{enumerate}

\subsection{Admissible Function Class}
In control applications, as seen in~\Cref{sec:hook}, the objective we wish to verify is rarely just a standalone neural network. More often, it is an expression that comprises several pieces: a policy $\pi(x)$, system dynamics $f(x,u)$, and some neural certificates $V(x)$. However, all these functions belong to a more generic function class named \textbf{computation graphs}. Intuitively, a computation graph is a directed acyclic graph (DAG) that represents a complicated function as a sequence of simple steps. Each node in the computation graph is an elementary operation (like addition, multiplication, or ReLU), and each edge passes the results of one step to the next. More formally, a computation graph can be defined as follows.

\vspace{0.05in}

\begin{definition}[Computation Graph]
A computation graph is a finite directed acyclic graph (DAG) $G=(V,E)$ whose
nodes $V=\{1,\dots,n\}$ are partitioned into input nodes $V_{\mathrm{in}}$,
parameter/constant nodes $V_{\mathrm{par}}$, and operator nodes $V_{\mathrm{op}}$.
Each node $i$ has an associated output vector $h_i\in\mathbb{R}^{d_i}$. For
each node $i$, let $u(i)=(u_1(i),\ldots,u_{m(i)}(i))$ be the ordered list of
its parents (so $(u_j(i),i)\in E$ and $m(i)=|u(i)|$).

\emph{Inputs.} For every input node $i\in V_{\mathrm{in}}$ we associate an
external input variable $x_i\in\mathbb{R}^{d_i}$, and the node simply
passes it through: $h_i=x_i$. If admissible sets $S_i\subseteq\mathbb{R}^{d_i}$
are specified for the inputs, the input space is
$S:=\prod_{i\in V_{\mathrm{in}}} S_i$, and we write
$x=(x_i)_{i\in V_{\mathrm{in}}}\in S$.

\emph{Parameters.} For every parameter/constant node $i\in V_{\mathrm{par}}$ a
fixed vector $c_i\in\mathbb{R}^{d_i}$ is given, and
$h_i=c_i$.

\emph{Operators.} For every operator node $i\in V_{\mathrm{op}}$ a primitive
map
\[
H_i:\ \prod_{k=1}^{m(i)}\mathbb{R}^{d_{u_k(i)}}\longrightarrow \mathbb{R}^{d_i}
\]
is specified (affine map, elementwise nonlinearity, concatenation, reduction,
etc.), and its output is
\[
h_i \;=\; H_i\!\big(h_{u_1(i)},\ldots,h_{u_{m(i)}(i)}\big).
\]

\emph{Outputs.} 
Let $O\subseteq V$ denote the set of output nodes (i.e., nodes with an out-degree of 0). The graph then
induces the function
\[
F_G:\ \prod_{i\in V_{\mathrm{in}}}\mathbb{R}^{d_i}\longrightarrow
\prod_{o\in O}\mathbb{R}^{d_o},\qquad F_G(x)=(h_o)_{o\in O}.
\]
\end{definition}

\vspace{0.05in}

\begin{example}[Toy Computation Graph]
    Suppose we want to represent the function
    \begin{equation}\label{eq:toygraph}
        f(x) = \mathrm{ReLU}(2x - 1) + 3.
    \end{equation}
    This can be written as a computation graph with four operator nodes: (1) multiply input by 2, (2) subtract 1, (3) apply ReLU, then (4) add 3. Each intermediate value is stored at a node, and the edges show how values flow from one step to the next. As shown in \Cref{fig:toygraph}, this small example illustrates the general idea: even a simple function can be broken down into basic operations, which is exactly how $\alpha,\!\beta$-CROWN reasons about much larger networks.
    \begin{figure}[h]
        \centering
        \includegraphics[width=\linewidth]{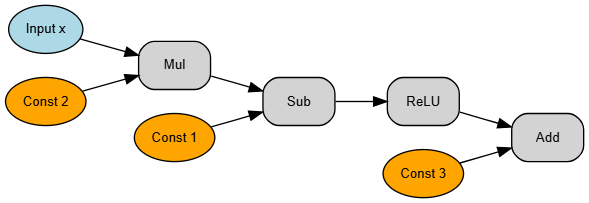}
        \caption{Computation graph of \Cref{eq:toygraph}.}
        \label{fig:toygraph}
    \end{figure}
\end{example}

The abstraction above captures many models and specifications that are commonly used in learning and control. For example, it is clear that any standard feed-forward network or more complicated neural network structures, such as residual networks or transformers, can be specified in this form with an appropriate choice of nodes, edges, and operators. Furthermore, many interesting properties of control systems, such as the Lyapunov condition, can be specified in this form. For example, for a discrete time system $x_{t+1} = f(x_t, \pi(x_t))$, the Lyapunov condition
\begin{align}
    F(x) = V(x) - V(f(x, \pi(x)))
\end{align}
is also a computation graph even when $\pi(x)$ and $V(x)$ are parametrized by neural networks (See~\Cref{fig:lyapunov_graph}).
\begin{figure}[h]
    \centering
    \includegraphics[width=\linewidth]{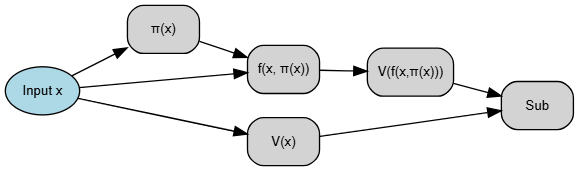}
    \caption{Computation graph of the Lyapunov condition. For illustration, $f$, $\pi$, and $V$ are abstracted as single nodes.}
    \label{fig:lyapunov_graph}
\end{figure}

\subsection{\texttt{auto\_LiRPA}: Automatic Linear Relaxation based Perturbation Analysis}\label{subsec:auto_lirpa}

% \huan{connect this to the example above.}

% \textcolor{blue}{Show high level intuition. Plot some linear bound figure (non-relu, make sure people understand we can handle others). One big computation graph with linear bound propagate.}

Suppose $F: \mathbb{R}^n \to \mathbb{R}$ is a computation graph (e.g., Fig.~\ref{fig:lyapunov_graph}) and verification domain is an $L_{\infty}$ box $\mathcal{B} = \prod_{i=1}^n [\ell_i, u_i]$. $\alpha,\!\beta$-CROWN efficiently verifies condition~\eqref{verify:condition} by propagating symbolic bounds on $F$ over domain $\mathcal{B}$ to get the linear lower bound of $F$ on $\mathcal{B}$, i.e., compute $\underline{a} \in \mathbb{R}^{n}, \underline{b} \in \mathbb{R}$ with
\begin{align}\label{veryfi:linear_bound}
    F(x) \geq \underline{a}^\top x + \underline{b}, \quad \forall x \in \mathcal{B}
\end{align}
We then claim verification successful if 
\begin{align}
    \inf_{x\in \mathcal{B}}\, \underline{a}^\top x + \underline{b} = (\underline{a}^+)^{\top}\ell + (\underline{a}^-)^{\top}u + \underline{b} > 0
\end{align}
holds with $\underline{a}^{+} = \max(\underline{a}, 0)$ and $\underline{a}^{-} = \min(\underline{a}, 0)$. To enable symbolic linear bound propagation with the full computation graph, we need to perform a \textbf{linear relaxation} of all primitive operators $H_i$, defined as follows.

\vspace{0.05in}

\begin{definition}[Linear relaxation]
Let $h:\prod_{j=1}^n \mathbb{R}^{p_j}\to\mathbb{R}^q$ and let the input domain be
a product box $Z=\prod_{j=1}^n [\ell_j,u_j]$, with $\ell_j,u_j\in\mathbb{R}^{p_j}$.
A \emph{linear relaxation} of $h$ on $Z$ is any collection of matrices and
vectors
\[
\big\{A_\ell^{(j)},A_u^{(j)}\in\mathbb{R}^{q\times p_j}\big\}_{j=1}^n,\quad
b_\ell,b_u\in\mathbb{R}^q
\]
such that for all $(x_1,\dots,x_n)\in Z$ the elementwise inequalities hold:
\[
\sum_{j=1}^n A_\ell^{(j)}x_j + b_\ell
\ \le\ h(x_1,\dots,x_n)\ \le\
\sum_{j=1}^n A_u^{(j)}x_j + b_u.
\]
Equivalently, stacking $x=\operatorname{col}(x_1,\dots,x_n)\in\mathbb{R}^{p}$
with $p=\sum_j p_j$ and similarly defining the global bounds $l = \operatorname{col}(\ell_1, \dots, \ell_n)$ and $u = \operatorname{col}(u_1, \dots, u_n)$ such that $Z=[\ell, u]$, a linear relaxation can be represented as a quadruple $(A_\ell,b_\ell,A_u,b_u)$ with $A_\ell,A_u\in\mathbb{R}^{q\times p}$,
$b_\ell,b_u\in\mathbb{R}^q$ satisfying
\[
A_\ell x + b_\ell \ \le\ h(x) \ \le\ A_u x + b_u.
\]
for all $x$ in the input domain, i.e., $x\in Z$.
\end{definition}

\vspace{0.05in}

Here, if $h$ is the primitive map of some node in a computation graph, the input box $Z$ is called the \textbf{preactivation bound} of this node. The tightness of the linear relaxation can be dependent on the preactivation bounds. Below, we give several examples of linear relaxation:
\paragraph{Affine map}
For $h(x)=Mx+c$ on any box $Z$, the relaxation is exact:
\begin{align}
A_\ell=M,\quad b_\ell=c,\qquad A_u=M,\quad b_u=c.
\end{align}
For multi-input $h(x^{(1)},\dots,x^{(r)})=\sum_{j=1}^r M_j x^{(j)}+c$,
use $A_{\ell}^{(j)}=A_{u}^{(j)}=M_j$ and $b_\ell=b_u=c$. This relaxation is exact.

\paragraph{Element-wise ReLU activation}
Let $h(x)=\mathrm{ReLU}(x)$ with $x\in[\ell,u]$. The linear relaxation $(A_\ell, b_\ell, A_u, b_u)$ can be chosen as
\begin{align}\label{relax:relu}
\text{LR}(h;[\ell,u]) =
\begin{cases}
(0,0,0,0),
& (u\le 0),\\
(1,0,1,0),
& (\ell\ge 0),\\
(\alpha,\,0,\,\frac{u}{u-\ell},\,-\frac{u}{u-\ell}\,\ell)
& (\ell<0<u)
\end{cases}
\end{align}
for any $\alpha\in[0,1]$. 
% We refer to the $u \leq 0$ and $\ell \geq 0$ case as \textbf{stable}, and the final case as \textbf{unstable}.

\paragraph{Sin} Let $h(x) = \sin(x)$ with $x \in [\ell,u]$. Because the sine function changes convexity, the linear relaxation must be adjusted based on the range of $\ell$ and $u$. A visualization of the linear bound can be found in Figure~\ref{fig:sin_bound}. Concretely, for an interval $[\ell,u] \subset [-\pi, \pi]$, the linear relaxation is constructed as follows (relaxations for other regions are derived similarly):
\begin{itemize}
    \item \textbf{Concave region ($\ell \ge 0$):} We use the secant line over $[\ell,u]$ as a valid lower bound, and an appropriate tangent line as an upper bound.
    \item \textbf{Convex region ($u \le 0$):} We use the secant line over $[\ell,u]$ as a valid upper bound, and an appropriate tangent line as a lower bound.
    \item \textbf{Mixed region ($\ell < 0 < u$):} Since the interval spans both convex and concave parts, we construct both bounds using tangent lines. 
    % \xiangru{Actually, it's possible that even in this case the direct secant line is still a valid bound, so it's not necessarily the tangent line. But, I checked the original CROWN paper, and this case is not discussed there, either.}
\end{itemize}

\begin{figure}[t]
    \centering
    \includegraphics[width=\linewidth]{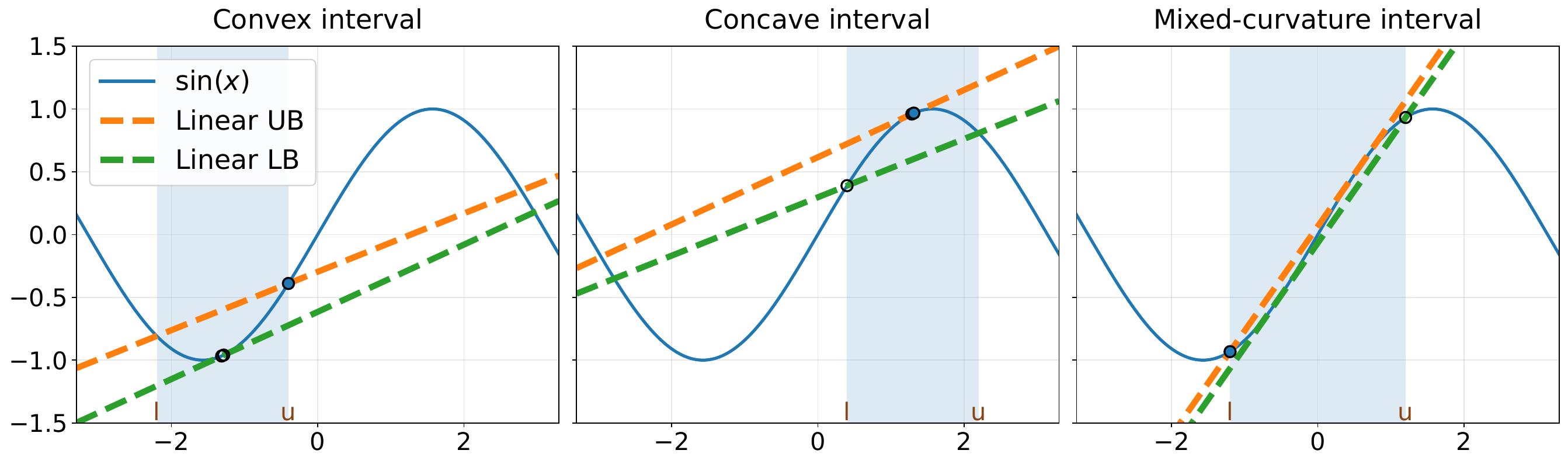}
    \caption{Linear lower and upper bound of $h(x) = \sin(x)$.}
    \label{fig:sin_bound}
\end{figure}

\paragraph{Element-wise Multiplication}
Let $h(x,y)=xy$ with $(x,y)\in[\ell_x,u_x]\times[\ell_y,u_y]$. The optimal linear relaxations are given by the McCormick envelopes, which establish two valid lower bounds and two valid upper bounds
\begin{equation}
    \begin{aligned}
        xy &\geq \ell_y x + \ell_x y - \ell_x \ell_y \\
        xy &\geq u_y x + u_x y - u_x u_y \\
        xy &\leq \ell_y x + u_x y - u_x \ell_y \\
        xy &\leq u_y x + \ell_x y - \ell_x u_y
    \end{aligned}
\end{equation}
Any convex combination of the lower bounds forms a valid lower relaxation, and any convex combination of the upper bounds forms a valid upper relaxation. For instance, selecting the first lower bound and the first upper bound yields the linear relaxation
\begin{equation}
    \mathrm{LR}(h;[\ell_x,u_x]\times[\ell_y,u_y]) = \left( [ \ell_y \,\,\,\ell_x], -\ell_x\ell_y, [\ell_y \,\,\, u_x] , -u_x\ell_y \right).
\end{equation}

We are now ready to introduce the actual backward symbolic bound propagation algorithm used in $\alpha,\!\beta$-CROWN. As a motivating example of how the algorithm works to derive a bound like in equation~\eqref{veryfi:linear_bound}, we consider a simple feed-forward neural network defined recursively as
\begin{equation}
\begin{aligned}
    F(x) &= c^\top z^{n}(x) \\
    z^{k+1}(x) &= \mathrm{ReLU}(W^{k+1}z^{k}(x)) \\
    z^{0}(x) &= x.
\end{aligned}
\end{equation}

Assume we have a verification target of $F(x) > 0$, and that all intermediate preactivation bounds have already been computed. We start the propagation with the exact equation $F(x) = c^\top z^n(x)$.
To propagate backward through the final ReLU layer, let $y^n(x) = W^n z^{n-1}(x)$ be the preactivation. From the linear relaxation in~\cref{relax:relu}, we can construct diagonal matrices $A_\ell^n, A_u^n$ and bias vectors $b_\ell^n, b_u^n$ such that
\begin{equation}
    A_\ell^n y^n(x) + b_\ell^n \leq \mathrm{ReLU}(y^n(x)) \leq A_u^n y^n(x) + b_u^n.
\end{equation}
Because $F(x)$ is a linear combination of these ReLU outputs weighted by $c$, we must split $c$ into its positive and negative components element-wise: $c^+ = \max(c, 0)$ and $c^- = \min(c, 0)$. Multiplying the relaxation by these components gives us valid bounds in terms of the preactivation $y^n(x)$
\begin{equation}
    \underline{A}_{\text{pre}}^{n} y^n(x) + \underline{b}_{\text{pre}}^{n} \leq F(x) \leq \overline{A}_{\text{pre}}^{n} y^n(x) + \overline{b}_{\text{pre}}^{n},
\end{equation}
where
\begin{equation}
\begin{aligned}
\overline{A}_{\text{pre}}^{n} &= (c^+)^\top A_u^n + (c^-)^\top A_\ell^n, \\
\overline{b}_{\text{pre}}^{n} &= (c^+)^\top b_u^n + (c^-)^\top b_\ell^n, \\
\underline{A}_{\text{pre}}^{n} &= (c^+)^\top A_\ell^n + (c^-)^\top A_u^n, \\
\underline{b}_{\text{pre}}^{n} &= (c^+)^\top b_\ell^n + (c^-)^\top b_u^n.
\end{aligned}
\end{equation}
Next, we substitute $y^n(x) = W^n z^{n-1}(x)$ to step backward through the linear layer. This yields the bounds in terms of the previous layer's post-activation $z^{n-1}(x)$:
\begin{equation}
\underline{A}^{n-1} z^{n-1}(x) + \underline{b}^{n-1} \leq F(x) \leq \overline{A}^{n-1} z^{n-1}(x) + \overline{b}^{n-1}
\end{equation}
where $\overline{A}^{n-1} = \overline{A}_{\text{pre}}^{n} W^n$ and $\underline{A}^{n-1} = \underline{A}_{\text{pre}}^{n} W^n$, while the biases remain unchanged ($\overline{b}^{n-1} = \overline{b}_{\text{pre}}^{n}$ and $\underline{b}^{n-1} = \underline{b}_{\text{pre}}^{n}$).

Now we can substitute $z^{n-1}(x)$ with $\mathrm{ReLU}(W^{n-1}z^{n-2}(x))$ and recursively continue this propagation until we reach the input $x$. This backward linear bound propagation naturally preserves the linear dependencies between nodes. As a result, it produces much tighter bounds compared to plain interval bound propagation (IBP), which naively pushes independent ranges forward. A full formal algorithm of this procedure is presented in \Cref{alg:backward-upper}. For simplicity of notation, we assume that there is only one output node and one input node, but the algorithm extends to multiple input/output nodes trivially. This algorithm can terminate with theoretical guarantees as follows.

\begin{algorithm}[t]
\caption{Backward Linear bound propagation}
\label{alg:backward-upper}
\begin{algorithmic}[1]
\Require Computation Graph $G=(V,E)$ of $F:\mathbb{R}^n\!\to\!\mathbb{R}$; for each non-input node $i$ a valid local relaxation $(A_{\ell,i}^{(j)},b_{\ell,i},A_{u,i}^{(j)},b_{u,i})$ w.r.t. parents $j\in u(i)$ given intermediate bounds and the input domain $\mathcal{B}$; output node $o$; input node $x$
\Ensure $\underline{A}\in\mathbb{R}^{1\times n}$, $\underline{b}\in\mathbb{R}$ with $F(x)\ge \underline{A}x+\underline{b}$ on $\mathcal B$
\State Set $\underline{A}_o\gets 1$, $\underline{d}\gets 0$; $\underline{A}_i\gets 0$ for all $i\neq o$
\State $deg_i\gets\textproc{GetOutDegree}(G)$ 
\State
initialize queue $Q\gets[o]$
\While{$Q\neq\emptyset$}
  \State $i\gets Q.\textproc{pop}()$
  \ForAll{$j\in u(i)$}
     \State $\Lambda_j \gets (\underline{A}_i)^{+} A_{\ell,i}^{(j)} + (\underline{A}_i)^{-} A_{u,i}^{(j)}$
     \State $\underline{A}_j \gets \underline{A}_j + \Lambda_j$
     \State $deg_j \gets deg_j - 1$
  \EndFor
  \State $\underline{d}\gets \underline{d} + (\underline{A}_i)^{+} b_{\ell,i} + (\underline{A}_i)^{-} b_{u,i}$
  \State $\underline{A}_i\gets 0$
  \ForAll{$j\in u(i)$}
    \If{$deg_j=0$ \textbf{and} $j$ is not input}
      \State $Q.\textproc{push}(j)$
    \EndIf
  \EndFor
\EndWhile
\State \textbf{return} $\underline{A}\gets \underline{A}_x$, $\underline{b}\gets \underline{d}$
\end{algorithmic}
\end{algorithm}

\begin{theorem}
    When Algorithm~\ref{alg:backward-upper} terminates, it is guaranteed that  we have
    \begin{align}
        F(x) \geq \underline{A}x + \underline{b}, \quad \forall x \in \mathcal{B}
    \end{align}
    and thus we obtain a sound lower bound of $F(x)$.
\end{theorem}

For a proof of this theorem, we refer the reader to Theorem 1 in the original \texttt{auto\_liRPA} paper~\cite{xu2020automatic}. The algorithm requires sound linear relaxation, which in turn requires access to preactivation bounds. We note that these preactivation bounds can be obtained in several ways, where the simplest and tightest approach is to recursively apply Algorithm~\ref{alg:backward-upper}. To be more precise, for each node where we need its preactivation bound, we apply Algorithm~\ref{alg:backward-upper} to the subgraph of $G$ that has this specific node as output node. Since for each input node we already have its preactivation bound, i.e., the domain $\mathcal{B}$, the recursion will eventually terminate. 

Intuitively, backward bound propagation is analogous to gradient backpropagation: both traverse the computation graph and repeatedly propagate a local quantity at each node. The key difference is that gradient backpropagation propagates derivatives, whereas LiRPA propagates linear upper/lower relaxations, composing them across layers to obtain a global affine bound. This propagation is computationally lightweight: its runtime is comparable to a standard backward gradient pass, simply applying a closed-form local linear relaxation in place of a gradient at each node. This efficiency makes it practical to tighten certificates by running the propagation over many sub-domains of $\mathcal{B}$ in parallel or batch (input domain branch-and-bound), as detailed next.

From Algorithm~\ref{alg:backward-upper}, it is evident that any nonlinear node can be handled if it admits a linear relaxation that can be efficiently computed. 
$\alpha,\!\beta$-CROWN implements many commonly used nonlinearities like $\mathrm{ReLU}, \tanh, *, /$, trig functions, and many others. 
% Moreover, once those nonlinearities can be handled, by the nature of the algorithm and structure of the computation graph, all the compositions of these functions can be handled efficiently. 
We list some currently implemented nonlinearities in Table~\ref{tab:supported_operators}.

\begin{table}[t]
\centering
\caption{Operators supported by $\alpha,\!\beta$-CROWN.}
\label{tab:supported_operators}
\renewcommand{\arraystretch}{1.15}
\setlength{\tabcolsep}{4pt}
\begin{tabular}{p{0.28\columnwidth} p{0.62\columnwidth}}
\toprule
\bfseries Category & \bfseries Operators \\
\midrule
Element-wise arithmetic \& math
& add, sub, mul, div; neg, pow, abs, square, reciprocal, exp, log, sqrt; sin, cos, tan, atan \\
\midrule
Nonlinear activations
& relu, hardtanh; tanh, sigmoid, gelu \\
\midrule
Reduction
& reduce\_sum, reduce\_mean, reduce\_max; max, min \\
\midrule
Linear / Matrix
& linear; matmul \\
\midrule
Convolutional
& conv2d; conv\_transpose2d \\
\midrule
Shape manipulation
& reshape, flatten; transpose; squeeze, unsqueeze; expand; concat, slice; gather \\
\midrule
Stochastic / Regularization
& dropout; softmax \\
\bottomrule
\end{tabular}
\end{table}

% \textcolor{blue}{Discuss differentiability of the bound instead of certified training}

\subsection{Handling Jacobian}\label{subsec:Jacobian computation}

% \huan{Move right after the discussion on CROWN}

In control problems, verification tasks often require bounding not just the output of a computation graph $G$, but also the Jacobian of its induced function $F_G$. For example, analyzing the Lyapunov stability for a continuous-time system requires bounding the gradient $\nabla V$ as a function of its input state $x$~\cite{li2025neural}. The key idea is to extend the graph itself so that it explicitly encodes the gradient computation via the chain rule, just like the auto-differentiation. Starting from the output node, we backtrack through the graph, attaching a corresponding ``gradient node" to each operator by taking local derivatives. This augments the original computation graph with a new set of nodes and edges, forming an extended computation graph that represents the backward gradient propagation. Evaluating this extended graph directly yields the Jacobian with respect to the inputs of the original graph. The same linear relaxation-based techniques implemented in auto\_LiRPA can then be applied to this extended graph, thereby producing certified bounds on the Jacobian in exactly the same manner as bounds on the original function. This procedure allows us to treat Jacobian computation as essentially another operator on the computation graph that can be handled easily.

\subsection{Branch and Bound (BaB)}\label{sec:bab}

% \textcolor{blue}{Introduce upper bound, so that it naturally leads to optimization problem}

The linear bound obtained by bound propagation may be potentially conservative, in which case the affine bound may sit noticeably below $F$ on the domain $\mathcal{B}$, and thus misses some verifiable cases. To mitigate this issue, $\alpha,\!\beta$-CROWN can perform \textbf{input domain branch-and-bound} to dynamically partition the input domain $\mathcal{B}$ into many small pieces $\{\mathcal{B}_i\}_{i=1}^n$. It then runs Algorithm~\ref{alg:backward-upper} to get a lower bound of $F$ on each of the subdomains, and uses the minimum of these lower bounds as the final certificate. Intuitively, since a function could be better approximated by linear functions under a smaller input domain, the bound we get from Algorithm~\ref{alg:backward-upper} on each of the subdomains will be tighter. Furthermore, for a piecewise linear activation like $\mathrm{ReLU}$, partitioning the domain significantly increases the likelihood of neuron stabilization (as seen in Equation~\eqref{relax:relu}). The full algorithm of input branch-and-bound is detailed in Algorithm~\ref{alg:branch-and-bound}.

\begin{algorithm}[h]
\caption{Input Domain Branch-and-Bound Verification}
\label{alg:branch-and-bound}
\begin{algorithmic}[1]
\Require function $F$, initial domain $\Omega$, batch size $B$
\State $\mathcal{S} \gets \textsc{Stack}([\Omega])$ \Comment{unverified subdomains}
\While{$\mathcal{S} \neq \emptyset$}
  \State $\mathcal{B} \gets \Call{PopBatch}{\mathcal{S}, B}$ \Comment{pop up to $B$ subdomains}
  \State $\mathcal{U} \gets \emptyset$
  \For{each $D \in \mathcal{B}$}
    \State compute lower bounds of $F(x)$ over $D$
    \If{$F(x) > 0$ cannot be certified on $D$}
      \State add $D$ to $\mathcal{U}$
    \EndIf
  \EndFor
  \If{$\mathcal{U} \neq \emptyset$}
    \For{each $D \in \mathcal{U}$}
        \State $D_1, D_2$ $\gets$ \Call{Split}{D}; \Call{Push}{$\mathcal{S}, D_1, D_2$}
    \EndFor
  \EndIf
\EndWhile
\If{$\mathcal{S} = \emptyset$}
  \State \Return \textsc{Verified} \Comment{all subdomains verified}
\EndIf
\end{algorithmic}
\end{algorithm}

In practice, since bound propagation is very efficient and parallelizable, the efficacy of the BaB algorithm hinges on how we split each unverified domain in lines 10 and 11 of Algorithm~\ref{alg:branch-and-bound}. While partitioning the domain along any dimension shrinks the input space and generally improves the linear relaxation, splitting along different dimensions can yield vastly different improvements in the resulting lower bound. To this end, $\alpha,\!\beta$-CROWN implements several different heuristics for the $\text{SPLIT}$ function in line 11. Here we introduce the two most commonly used heuristics~\cite{bunel2018unified}.

\paragraph{Naive Branching} In this splitting heuristic, we simply split the longest edge of the input domain $\mathcal{B}$. To be more precise, suppose $\mathcal{B}= \prod_{i=1}^n [\ell_i, u_i]$. We first compute $k = \arg \max_i |u_i -\ell_i|$, and split $\mathcal{B}$ into 
\begin{equation}
\begin{aligned}
    \mathcal{B}_1 &= \prod_{i\neq k} [\ell_i, u_i] \times [\ell_i, \frac{\ell_i + u_i}{2}] \\
    \mathcal{B}_2 &= \prod_{i\neq k} [\ell_i, u_i] \times [ \frac{\ell_i + u_i}{2}, u_i]
\end{aligned}
\end{equation}
\paragraph{Smart Branching~\cite{bunel2018unified}} This differs from the naive branching only in how the split index $k$ is chosen. On the domain $\mathcal{B} = \prod_{i=1}^n [\ell_i, u_i]$, after the bound propagation, we get a sound affine lower bound of $F$ satisfying
\begin{align}
    F(x) \ge \underline{a}^\top x + \underline{b}, \quad \forall x \in \mathcal{B}
\end{align}
We split the index $j$ that satisfies
\begin{align}
    j = \arg \max_i |\underline{a}_i|(u_i - \ell_i),
\end{align}
which essentially computes the most sensitive dimension that could potentially improve the bound.

In practice, we often prefer smart branching over naive branching, although it does not always guarantee improvements. Thanks to the efficiency of the bound-propagation algorithm~\eqref{alg:backward-upper}, which can be done on GPUs, we can acceptably divide the domain $\mathcal{B}$ into tens of thousands of pieces. Furthermore, $\alpha,\!\beta$-CROWN implements efficient parallelization that effectively computes bounds on these divided domains simultaneously. After repeated branching, each domain will be tiny, and the linear over-approximation will often be tight enough to enable verification. This provides the main intuitive idea on how $\alpha,\!\beta$-CROWN achieves state-of-the-art verification results on problem~\eqref{verify:condition}.

% \textcolor{blue}{Introduce what the user should do if verification is unsuccessful. For each feature, briefly discuss and cite, mention how user can open it}

\subsection{Supporting General Satisfiability Problems}\label{subsec:general_spec}

% \huan{Connect to the Lyapunov example}

Up to now, we have focused on obtaining tight numeric and linear bounds for formulas like~\eqref{verify:condition}. In practice, when we solve satisfiability problems, specifications are often more complex and involve logical combinations of multiple constraints, just like the motivating example of Lyapunov stability in Section~\ref{sec:hook}. Thanks to the parallelizable nature of bound propagation and the flexibility of the branch and bound procedure, $\alpha,\!\beta$-CROWN can naturally handle such cases. In principle, $\alpha,\!\beta$-CROWN can handle many cases that a traditional SMT solver for real numbers, such as dReal, can handle, and in a much more scalable fashion. Here we discuss how $\alpha,\!\beta$-CROWN handles the following three typical forms.

\subsubsection{Multiple OR constraints}
An OR specification requires at least one of several conditions to hold:
\begin{equation}
    F_1(x) > 0 \lor F_2(x) > 0\lor \dots\lor F_k(x) > 0.
\end{equation}
To encode this efficiently, we exploit the fact that all operations in $\alpha,\!\beta$-CROWN are vectorized. A linear transformation on the final layer can re-parameterize the outputs into 
\begin{equation}
    \Tilde{F}(x) \succ_{\lor} 0,
\end{equation}
where verification succeeds if any entry of $\Tilde{F}(x)$ is positive (i.e., $\max_i \tilde{F}_i(x) > 0$). In other words, all disjunctions are checked in parallel as components of a tensor.
\begin{example}[Implication]
Suppose we want to verify a simple implication: $A(x) \to B(x)$. For instance:
\begin{itemize}
    \item \emph{If} the state is inside a ``trigger region'' $A(x): \Vert x\Vert \; \leq 1$,
    \item \emph{then} the system must satisfy a safety condition $B(x): g(x) \leq 0$.
\end{itemize}
This is equivalent to checking
\begin{equation}
    (\Vert x\Vert > 1)\lor (g(x) \leq 0).
\end{equation}
In $\alpha,\!\beta$-CROWN, this can be written in the form as
\begin{equation}
    \tilde{F}(x) = \begin{bmatrix}
        \Vert x\Vert - 1 \\
        -g(x)
    \end{bmatrix} \succ_{\lor} 0
\end{equation}
The bound of $\tilde{F}(x)$ above is computed in vector form. At any stage of the algorithm, if for a subdomain $\mathcal{B}$ the propagated bound already certifies $\tilde{F}\succ_{\lor} 0$, then the subdomain is considered verified. A real example of the implication type constraints is given in Section~\ref{sec:control_app}.
\end{example}

\subsubsection{Multiple AND constraints}
An AND specification requires \emph{all} constraints to hold simultaneously:
\begin{equation}
    F_1(x) > 0 \land F_2(x) > 0 \land \dots \land F_k(x) > 0.
\end{equation}
In $\alpha,\!\beta$-CROWN, each conjunct is considered separately. During BaB, all of them are added to the domain stack $\mathcal{S}$. Verification succeeds only when every conjunct has been individually certified (i.e., removed from the stack). The example in Section~\ref{sec:hook} is a real application of multiple AND constraints. 

\subsubsection{General Boolean specification}
More complex specifications may mix AND and OR operators. For example,
\begin{equation}
    (F_1(x)> 0 \lor F_2(x)> 0) \land (F_3(x) > 0) \lor (F_4(x) > 0).
\end{equation}
Such formulas can always be rewritten in conjunctive normal form (CNF): a conjunction of clauses, each of which is a disjunction. Each disjunctive clause can be handled as in case (1), and the outer conjunction is treated as in case (2).
% \footnote{In practice, when the specification is provided in the standard \texttt{vnnlib} format, what is passed to the verifier is the \emph{satisfiability condition}, i.e., the negation of the property we aim to verify. As a result, the specification is expressed in disjunctive normal form (DNF).}.

\subsection{Overall Pipeline for Solving Satisfiability Problems}
The $\alpha,\!\beta$-CROWN framework is designed to be flexible. In principle, multiple verifiers and falsifiers can be composed in different orders, depending on the verification target and efficiency considerations. However, in typical control problems, we often follow a practical pipeline that balances speed and completeness. The pipeline is illustrated in Figure~\ref{fig:abcrown}.

\begin{figure}[htb]
    \centering
    \includegraphics[width=\linewidth]{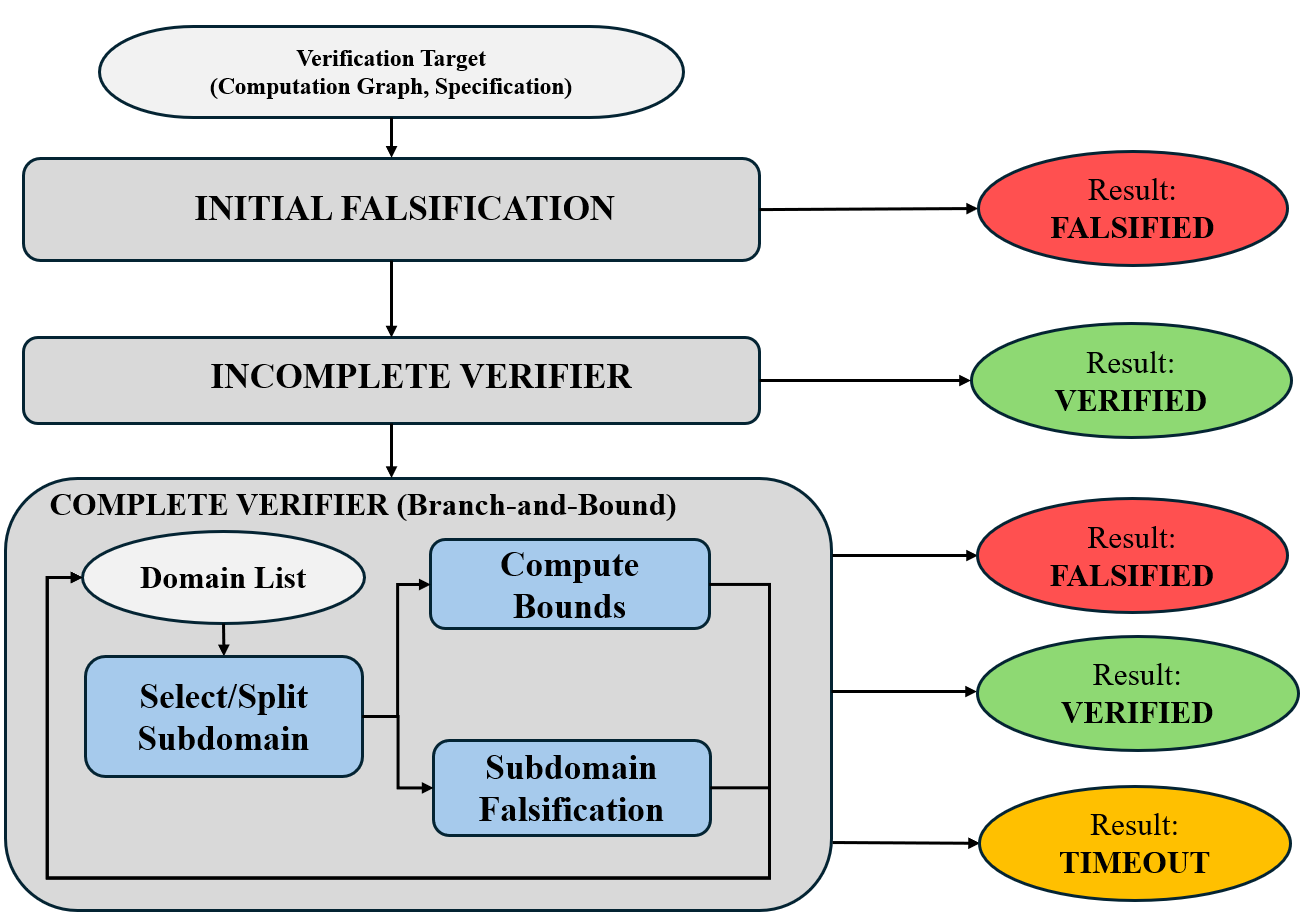}
    \caption{Overview structure of $\alpha,\!\beta$-CROWN}
    \label{fig:abcrown}
\end{figure}
\subsubsection{Initial Falsification}
The process begins with a falsifier, where we typically apply a projected gradient descent (PGD)~\cite{madry2017towards} to search for counterexamples to the verification condition. Concretely, PGD performs gradient-based minimization on the violation objective (e.g., the margin in~\eqref{verify:condition}) over a large batch of initial candidates, and after each step projects the candidates back onto the allowed input set. If a violating input $x\in \mathcal{B}$ is found such that the specification fails, the procedure terminates immediately with falsification.

\subsubsection{Incomplete verifier (bound propagation with no BaB)}
If no counterexample is found, we proceed with an incomplete verifier, which consists of a simple bound propagation pass with no branch-and-bound. This step computes sound but possibly conservative bounds with auto\_LiRPA.
\begin{itemize}
    \item If the bounds already certify the property, the procedure stops here and returns verification.
    \item Otherwise, the bounds provide initialization for the complete verifier, for example, intermediate layer preactivation bounds.
\end{itemize}

\subsubsection{Complete verifier (branch-and-bound)}
Finally, if the incomplete verifier is inconclusive, we use a complete verifier based on input domain branch-and-bound as described in \Cref{sec:bab}. The input domain is recursively partitioned into subdomains, and sound affine bounds are propagated on each of them. Verification succeeds once every subdomain has been refined sufficiently so that its propagated bounds alone can certify the property. In the meantime, during branch-and-bound, since the domain becomes smaller, counterexamples might be easier to locate. Therefore, we also perform a more fine-grained falsification on each subdomain during branch-and-bound to ensure that no counterexamples are missed. In practice, branch-and-bound runs under some budget, like a predefined time limit. If the budget is exhausted before the verifier can prove or refute the property, an inconclusive result is returned.

\section{Applications to Control}\label{sec:control_app}

% \textcolor{blue}{clearly write prerequisite of abcrown, need local region, f needs some specific form}

In this section, we provide several applications of using $\alpha,\!\beta$-CROWN in the control setting. In particular, since casting control certificates into a satisfiability form is often the most subtle step, we focus on the corresponding problem reformulations. We shall discuss how to use it to verify the stability certificate, safety certificate, and contraction.

Throughout this section, we consider either a closed-loop discrete-time system:
\begin{align}\label{sys:discrete}
    x_{t+1} = f(x_t, \pi(x_t)) = g(x_t),
\end{align}
or a closed-loop continuous system
\begin{align}\label{sys:continuous}
    \dot{x} = f(x, \pi(x)) = g(x).
\end{align}
We denote the solution of this system with initial condition $x_0 = x$ as $\phi(t,x)$. In what follows, we assume that the controller $u(\cdot)$ and corresponding certifications, such as Lyapunov function or contraction metrics, are given and focus exclusively on verification. We defer all details on how to co-learn the controller with a formal certificate to the next section. Furthermore, we assume that $f,\pi$ are general computation graphs with nonlinearities being those presented in Table~\ref{tab:supported_operators} that are supported by $\alpha,\!\beta$-CROWN.

A central challenge in applying $\alpha,\!\beta$-CROWN to control tasks is to make the classical certificate checkable by a verifier that only supports box input domains. For example, since the traditional Lyapunov theory and barrier function theory need forward invariance, they are naturally stated on implicitly defined sets such as the sublevel set $\Omega_{\rho}^V = \{x: V(x) \leq \rho\}$, whereas $\alpha,\!\beta$-CROWN can only verify conditions in a box domain $\mathcal{B}$. The naive workaround, which is to verify the decrease condition over the entire box $\mathcal{B}$ and then locate the maximal invariance set inside it, is typically overly conservative and is likely to fail even when a valid certificate exists. Therefore, we instead rewrite each certificate as an implication on the sublevel set, and then convert it into an equivalent satisfiability condition over the surrounding box. In the following sections, we shall demonstrate in detail how to reformulate classical control certificates, including discrete-time and continuous-time Lyapunov analysis, continuous-time barrier function, and discrete-time contraction analysis, into verifier-friendly constraints over a box domain. 

% \huan{Motivate what are in general needed to use alpha-beta-CROWN? What are the assumptions? How does it differ from previous work?}

\subsection{Reachability Analysis}
Before we dive into the actual certification criterion with $\alpha,\!\beta$-CROWN, we shall first briefly discuss how to use the bound produced by auto\_LiRPA to conduct reachability analysis for a discrete-time system. Reachability analysis asks: given an initial set of states $X_0$ and a closed-loop dynamics map $g(x)=f(x,\pi(x))$, what set of states can the system reach after one or multiple time steps? The reachable set computed can be used by downstream tasks, such as safety verification over a finite horizon. In general, exactly computing $X_{t+1}=g(X_t)$ is intractable for nonlinear dynamics and neural-network controllers. Instead, we aim to compute a sound over-approximation $\widehat{X}_{t+1}\supseteq g(X_t)$ that is cheap to obtain but guaranteed to contain all successors. auto\_LiRPA provides exactly this capability: by treating $g$ as a computation graph and $X_t$ as an input domain (typically a box or a polytope), it can certify elementwise bounds on $g(x)$ for all $x\in X_t$, which directly yields a sound reachable set over-approximation, and iterating this one-step over-approximation produces a forward reachable tube. 

\subsection{Discrete-Time Lyapunov Analysis}\label{sec:discrete_Lyapunov}
The goal of stability analysis is to certify that the controller $\pi(\cdot)$ stabilizes the systems in \eqref{sys:discrete} and \eqref{sys:continuous}, and to characterize the domain of initial states from which trajectories are guaranteed to converge to the equilibrium. Formally, we aim to find the region of attraction (ROA) of the system, defined as follows.

\begin{definition}[Region of Attraction (ROA)~\cite{khalil2002nonlinear}]
    Region of attraction for the system~\eqref{sys:discrete} or~\eqref{sys:continuous} is the set $\mathcal{R}$ such that $\lim_{t\to \infty} \phi(t,x_0) = x^\ast$ for all $x_0 \in \mathcal{R}$. 
\end{definition}

In practice, Lyapunov theory is the standard tool to certify such guarantees: one seeks a scalar “energy” function $V$ whose value decreases along closed-loop trajectories, so that the trajectories from a forward invariant set converge to the equilibrium, yielding an inner approximation of the ROA~\cite{lyapunov1992general}. We now state a Lyapunov-based theorem that formalizes this certification criterion that can be used by $\alpha,\!\beta$-CROWN. We first start with discrete-time systems~\cite{yang2024lyapunov}. 

\begin{theorem}[Lyapunov Theorem (Discrete-Time)~\cite{bof2018lyapunov}]\label{thm:lyapunov}
Given a forward invariant set $\mathcal{S}$ containing the equilibrium, if there exists a function $V: \mathcal{S} \to \mathbb{R}$ such that we have
\begin{equation}
\begin{aligned}
V(x^{\ast})&=0\\
V(x)&>0\;\;(\forall\,x \in \mathcal S\setminus\{x^{\ast}\})\\
V\left(g(x_t)\right) - V(x_t)&\leq-\kappa V(x_t)\;\;(\forall\,x_t \in \mathcal S)
\end{aligned}
\end{equation}
for some $\kappa > 0$, then $\mathcal{S}$ is a subset of the region of attraction for the discrete-time system~\eqref{sys:discrete}.
\end{theorem}

Now suppose we have synthesized a candidate $V$ through learning or other means, and would like to formally verify its correctness with $\alpha,\!\beta$-CROWN. As the verifier can only handle a box input, we begin by fixing a box-shaped region of interest $\mathcal{B}$ containing the equilibrium state $x^\ast$ and constraining the forward invariant set $\mathcal{S}$ to be the intersection of the box $\mathcal{B}$ with a sublevel set of $V$, as
\begin{align}
    \mathcal{S} := \{x\in B\mid V(x) < \rho\}
\end{align}
for some $\rho > 0$ that ensures for any $x \in \mathcal{S}$, we have $x_{t+1} \in \mathcal{B}$. Clearly, by the sublevel set structure, if the Lyapunov function values are certified to decrease along the trajectory, $x_{t+1}$ will also land in $\mathcal{S}$, making $\mathcal{S}$ forward invariant. Therefore, we reach the following verification condition. 

\begin{theorybox}{Discrete-Time System Stability Verification}
\begin{theorem}\label{thm:discrete_stability_condition}
Suppose that $V(x^*) = 0$ and $V$ is positive everywhere else. Let $F(x)\!:=\!V(g(x))\!-\!(1-\kappa)V(x)$. If the condition
\begin{equation}\label{eq:discrete_stability_condition}
    (-F(x) \geq 0 \wedge g(x) \in \mathcal{B}) \vee (V(x) \geq \rho), \quad \forall x \in \mathcal{B}
\end{equation}
holds for some $\rho > 0$, then $\mathcal{S}$ is a certified ROA inner-approximation for the closed-loop system~\eqref{sys:discrete}.
\end{theorem}
\end{theorybox}

\begin{proof}
    By theorem~\ref{thm:lyapunov}, we just need to check that $\mathcal{S}$ is forward invariant and the Lyapunov function decreases along trajectories for any trajectories in $\mathcal{S}$. We take any $x \in \mathcal{S}$. Then $x \in \mathcal{B}$ and $V(x) < \rho$. Therefore, we know that $g(x) \in \mathcal{B}$ and $V(g(x)) \leq (1-\kappa)V(x) < \rho$. Therefore, $g(x) \in \mathcal{S}$ and $\mathcal{S}$ is forward invariant. The above analysis also shows that the Lyapunov function values decrease along trajectories in $\mathcal{S}$. 
\end{proof}

Now, since the verification condition~\eqref{eq:discrete_stability_condition} is on a box $\mathcal{B}$, and all the conditions can be cast into scalar value inequalities, it can be directly handled by $\alpha,\!\beta$-CROWN as described in the general specification section~\ref{subsec:general_spec}. 

\subsection{Continuous-Time Lyapunov Analysis}
For continuous-time systems, the essence is similar~\cite{li2025neural}. To present the verification criterion, we first state the continuous-time Lyapunov theorem.

\begin{theorem}[Lyapunov Theorem (Continuous-Time)~\cite{lyapunov1992general}]\label{thm:cont_lyapunov}
Given a forward invariant set $\mathcal{S}$ containing the equilibrium, if there exists a function $V: \mathcal{S} \to \mathbb{R}$ such that we have
\begin{equation}
\begin{aligned}
V(x^{\ast})&=0\\
V(x)&>0\;\;(\forall\,x \in \mathcal S\setminus\{x^{\ast}\})\\
\nabla V(x) f\bigl(x,\pi(x)\bigr)&\leq-\kappa V(x)\;\;(\forall\,x \in \mathcal S)
\end{aligned}
\end{equation}
for some $\kappa > 0$, then $\mathcal{S}$ is a subset of the region of attraction for the continuous-time system~\eqref{sys:continuous}.
\end{theorem}

We note that this differs from the discrete-time Lyapunov theorem in the sense that it requires the differentiability of the function $V$, and the decrease condition is now formalized with a differential change rather than a discrete change. This makes it more difficult to parametrize $V$ to satisfy the positive definiteness condition directly when we wish to use a neural network based parametrization. Therefore, in practice, we exclude the verification on a small neighborhood around the origin. Moreover, the forward invariance condition in continuous-time theory is trickier than in the discrete-time counterpart. We must ensure that, on the intersection of the ROA boundary and the box boundary, the closed-loop vector field points inward so that trajectories cannot escape the box. In conclusion, the verification condition can be encoded in the following theorem. 

\begin{theorem}[\cite{li2025two}]\label{continuous_stability_condition}
Let \(0 < c_1 < c_2\). We denote $V^{\leq c} := \{x: V(x) \leq c \}$ and $V^{\leq c_2\setminus \leq c_1} := V^{\leq c_2} \setminus V^{\leq c_1}$. Suppose we have
\begin{align}
  \label{eq:thm_violation}
  &x\in V^{\leq c_2\setminus \leq c_1} \cap \mathcal{B}
  \implies
  \nabla V(x)\cdot f\bigl(x,\pi(x)\bigr)\leq -\kappa V(x),\\
  \label{eq:thm_boundary}
  &x\in \partial\mathcal{B}\cap V^{\leq c_2}
  \implies
  f\bigl(x,\pi(x)\bigr)\cdot \vec{n}(x) \leq 0,
\end{align}
 where $\vec{n}(x)$ is the outer normal vector of $\partial\mathcal{B}$ at $x$. Then both $V^{\leq c_1}$ and $V^{\leq c_2}$ are \emph{forward invariant}, and every trajectory initialized in $V^{\leq c_2}$ enters the smaller set $V^{\leq c_1}$ in finite time.
\end{theorem}

To convert this condition into a form that $\alpha,\!\beta$-CROWN can verify, again, we need the verification to happen on the box domain $\mathcal{B}$. Formally, it can be converted into the following criterion.
\begin{theorybox}{Continuous-Time System Stability Verification}
\begin{theorem}
    Let $F(x):= \nabla V(x) \cdot f(x,\pi(x)) +\kappa V(x)$ and $G(x):= f(x,\pi(x))\cdot \vec{n}(x)$. The above theorem is equivalent to verifying
    \begin{align}
    (F(x) \leq 0) &\vee (V(x) > c_2) \vee (V(x) < c_1)
    \label{eq:continuous_stability_1}\\
    (G(x) \leq 0) &\vee (x \notin \partial\mathcal{B}) \vee (V(x) > c_2)
    \label{eq:continuous_stability_2}
    \end{align}
    for all $x \in \mathcal{B}$. Moreover, we shall see that since $\mathcal{B}$ is a box domain, the condition $x \in \partial \mathcal{B}$ can be easily converted into several numerical conditions. 
\end{theorem}
\end{theorybox}

Since in practice the sublevel set $\{x: V(x) \leq c_1\}$ is small, this theorem essentially proves that $\{x: V(x) \leq c_2\}$ is an inner approximation of the ROA. Now both conditions~\eqref{eq:continuous_stability_1} and~\eqref{eq:continuous_stability_2} can be verified by $\alpha,\!\beta$-CROWN using the strategy discussed in Section~\ref{subsec:general_spec}.

\subsection{Robust Control under Disturbance and Uncertainty}

In this section, we consider generalized systems with exogenous inputs 
\begin{align}\label{eq:GenSys}
    x_{t+1} = f(x_t, \pi(x_t), w_t) = g(x_t, w_t).
\end{align}
We wish to find a robust forward invariant set under all possible disturbances, and to achieve this, we would need a generalized Lyapunov function. Let's consider the case where the disturbances $w_t$ are bounded pointwise by $\phi(w_t) \leq \nu$. In this case, the following theorem provides a sufficient condition that is compatible with $\alpha,\!\beta$-CROWN. 
\begin{theorem}[Uniform Disturbance Robustness]\label{thm:pointwise_robust_roa}
Let $\mathcal{B} \times \mathcal{B}^w\subset \mathbb{R}^n \times\mathbb{R}^{n_w}$ be a box constrained set containing $(0,0)$ and \(V : \mathbb{R}^{n} \to \mathbb{R}\) be a continuous positive-definite function. Suppose that the admissible set of disturbances given by a constraint function $\psi :\mathbb{R}^{n_w} \rightarrow \mathbb{R}$ and $\nu > 0$ is non-empty
\begin{align}
\mathcal{W} := \{\mathbf{w} : w_k \in \mathcal{B}^w, \,\, \psi(w_k) \leq \nu, \, \,\forall k\}
\end{align}
which defines the following product domain:
\begin{align}
    \mathcal{S}^V_\rho := \{ (x,w) \in (\Omega^V_\rho \cap \mathcal{B})\times \mathcal{B}^w : \psi(w) \leq \nu \}
\end{align}
and suppose that the following holds for some $\kappa \in (0,1)$:
\begin{equation}\label{eq:pointwise_condition}
\begin{aligned}
\mathcal{S}^V_\rho \subseteq \{(x,w) : f(x,w) \in \mathcal{B},\,\, F_V(x,w) \leq 0\},
\end{aligned}
\end{equation}
where $F_V(x,w) = V(f(x,w)) - (1-\kappa)V(x) - \psi(w)$. If $\nu/\kappa \leq \rho$, $x_0 \in \mathcal{X}:=\Omega^V_\rho \cap \mathcal{B}$ and $\mathbf{w}\in \mathcal{W}$, then the solution to~\eqref{eq:GenSys} satisfies $(x_k, w_k) \in \mathcal{S}^V_\rho$ for all $k$. Furthermore, if there is no disturbance and $\mathbf{w}=0$ and $\psi(0) \leq 0$, then $\Omega^V_\rho \cap \mathcal{B}$ is a forward-invariant subset of the ROA.
\end{theorem}

This result shows that if the change in \(V(x_k)\) can be bounded at each step by a known function \(\psi(w_k)\), and if this bound stays within a fraction \(\kappa\) of the energy level \(\rho\), then the trajectory remains within the sublevel set \(\Omega^V_\rho\). Given a candidate storage function $V$ and the constraint function $\psi$, Theorem \ref{thm:pointwise_robust_roa} can be directly verified by $\alpha,\!\beta$-CROWN with the following specification:
\begin{theorybox}{Discrete-Time systems Stability Verification with Uniform Bounded Disturbance}
\begin{equation}\label{eq:ver1}
\begin{split}
\bigl(F_V(x,w) \le 0 \land f(&x,w) \in \mathcal{B}\bigr)\lor\, (V(x) > \rho) \\ &{}\lor (\psi(w) > \nu)
\end{split}
\end{equation}
\end{theorybox}
\noindent
by checking all $(x,w) \in \mathcal{B}^x \times \mathcal{B}^w$. A similar theorem can be derived to handle integral-constrained disturbances.

\subsection{Contraction Metrics}

% \huan{Also connect to the general assumption of using alpha-beta-CROWN}

We further discuss how to convert the verification criterion for contractions into forms that can be tackled by $\alpha,\!\beta$-CROWN. Traditionally, discrete-time contraction conditions are usually expressed as matrix inequalities~\cite{tran2018convergence}, which are difficult for scalar-valued verifiers to handle and also rely on differentiability of the dynamics. Therefore, we propose to reformulate the contraction condition into a scalar satisfiability problem as follows.

% \begin{theorem}[~\cite{tran2018convergence}]\label{thm:matrix_version}
%     The system~\eqref{sys:discrete} is contracting in $\mathcal{X}$ if there exists a nonsingular matrix-valued function $\Theta: \mathcal{X} \to \mathbb{R}^{n\times n}$ and positive constants $\mu, \eta, \rho$ such that for all $x \in \mathcal{X}$ we have
%     \begin{equation}
%         \eta I \leq \Theta(x)^\top\Theta(x) \preceq \rho I, \quad F(x)^\top F(x) - I \preceq -\mu I
%     \end{equation}
%     where $F(x)= \Theta(f(x))\frac{\partial f}{\partial x}(x) \Theta^{-1}(x)$.
% \end{theorem}

% This criterion has several drawbacks. First of all, the verification of matrix inequalities can be hard for verifiers. The most common way of verifying matrix inequalities relies on Sylvester's criterion, as used by~\cite{10714396}. As can be seen from theorem~\ref{thm:matrix_version}, we need to verify multiple semi-definite conditions, and each of them requires us to verify a determinant condition for each of the $2^{n} - 1$ principal minors. Therefore, the verification problem cannot scale. Secondly, this condition requires the system to be continuously differentiable, as it requires the Jacobian of the dynamics. However, currently, many state-of-the-art results have controllers synthesized by ReLU networks~\cite{yang2013neural,wu2023neural}, which makes this theorem inapplicable. Therefore, we devise a new contraction condition with only scalar value conditions, as follows:

\begin{theorem}[Contraction~\cite{li2025neural}]
Assume that $\mathcal{X}$ is an open, connected, and forward invariant subset in $\mathbb{R}^n$, and the dynamics $f$ is continuous. Given a positive threshold $\epsilon > 0$, suppose there exists uniformly continuous $M(x) : \mathcal{X} \to \mathbb{S}^{n\times n}_{++}$ with a uniform lower bound $\mu I \preceq M(x)$ that satisfies 
\begin{equation}\label{verification}
\begin{aligned}
     &\ \sqrt{(f(x) - f(y))^\top M(f(x))(f(x)-f(y))}  \\ 
     \leq&\  \rho \sqrt{(x-y)^\top M(x)(x-y)} 
\end{aligned}
\end{equation}
for all $x \in \mathcal{X}$ and $y \in B(x; \epsilon) \cap \mathcal{X}$ and some $0< \rho < 1$, then the system~\eqref{sys:discrete} is contracting with respect to the Riemannian metric $d$ induced by $M$, i.e, the system satisfies
\begin{align}
    d(f(x), f(y)) \leq \rho\,d(x,y)
\end{align}
for any $x, y \in \mathcal{X}$. 
\end{theorem}

This condition can now be written in a form that is compatible with $\alpha,\!\beta$-CROWN, and the assumption is reduced to only assuming the continuity of the closed-loop dynamics, which can incorporate many neural-controlled systems. The forward invariant domain $\mathcal{X}$ can be found with Lyapunov sublevel sets or through condition~\eqref{eq:discrete_stability_condition}. Assuming we have already obtained such an invariant set, the contraction condition can be formulated as follows. We note that here the technique to transform the verification condition on the complex forward invariant set into a box is essentially the same as in the Lyapunov setting. Formally, the satisfiability problem to be solved can be defined as follows.

\begin{theorybox}{Discrete-Time System Contraction Analysis}
\begin{theorem}
Suppose we have already formally verified the forward invariance of
\[
\mathcal{X}:=\{x:V(x)<\rho\}\cap\mathcal{B}.
\]
For \((x,\delta)\in\mathcal{B}\times[-\epsilon,\epsilon]^n\), define
\[
\Delta_f(x,\delta):=f(x+\delta)-f(x),
\]
and
\begin{equation*}
G(x,\delta):=
\Delta_f(x,\delta)^\top M(f(x))\,\Delta_f(x,\delta)
-\rho^2\,\delta^\top M(x)\delta.
\end{equation*}
Assume that, for all \((x,\delta)\in\mathcal{B}\times[-\epsilon,\epsilon]^n\),
\begin{equation}\label{eq:contraction_condition}
\begin{aligned}
\bigl(x+\delta\in\mathcal{B}\bigr)\wedge\bigl(V(x)<\rho\bigr)\wedge\bigl(V(x+\delta)<\rho\bigr)
&\\
\implies G(x,\delta)\le 0.
\end{aligned}
\end{equation}
Then the system~\eqref{sys:discrete} is certifiably contracting on \(\mathcal{X}\) with contraction rate \(\rho'\in(\rho,1)\).
\end{theorem}
\end{theorybox}

\subsection{Barrier Function}
We further consider a safety certificate with barrier functions~\cite{ames2016control}. In this setting, we wish to prove the forward invariance of a safe set $C$ of the following form:
\begin{align}
    C &= \{x\in \mathbb{R}^n: h(x) \geq 0\}, \\
    \partial C &= \{x\in \mathbb{R}^n: h(x) = 0 \}, \\
    \text{Int}(C) &= \{x\in \mathbb{R}^n: h(x) > 0\},
\end{align}
where $h$ is a continuously differentiable function. The forward invariance of $C$ can then be certified if $h$ is a zeroing barrier function on some box $\mathcal{D}$ with $C \subseteq \mathcal{D}$, i.e. it satisfies:
\begin{align}
    \nabla h(x) \cdot f(x,\pi(x)) \geq -\alpha\,h(x), \quad \forall x \in \mathcal{D} \cap \{h(x) \geq 0\}
\end{align}
for some $\alpha > 0$. This condition can then be directly verified by $\alpha,\!\beta$-CROWN. More generally, we consider the control barrier function setting. Specifically for this setting, we consider a control-affine system:
\begin{align}
    \dot{x} = f(x) + g(x)u,
\end{align}
where the $u$ can take values from a compact convex polytope $U \subset \mathbb{R}^m$. $h$ is called a control barrier function if it satisfies 
\begin{align}
    \sup_{u \in U}\,[\nabla h(x)\cdot f(x) + (\nabla h(x) \cdot g(x))u ] \geq -\alpha h(x), 
\end{align}
$\forall x \in \mathcal{D} \cap \{h(x) \geq 0\}$. As the left-hand side is linear in $u$, the supremum over $u$ can be converted to a maximum over the vertices of the convex polytope $U$, which can then be directly verified by $\alpha,\!\beta$-CROWN. Formally, we get the following verification criterion:
\begin{theorybox}{Continuous-Time Barrier Function Verification}
\begin{theorem}
Let \(\mathcal D \supseteq C\) be a box, and let \(\mathcal V(U)\) denote the set of vertices of \(U\).
Define
\[
B_h(x)
:= \max_{u\in \mathcal V(U)}
\Bigl[\nabla h(x)^\top f(x) + \bigl(\nabla h(x)^\top g(x)\bigr)u\Bigr]
+ \alpha h(x).
\]
Then \(h:\mathcal D\to\mathbb R\) is a control barrier function if and only if, for all \(x\in\mathcal D\),
\begin{equation}\label{eq:cbf_verify}
B_h(x)\ge 0 \;\vee\; h(x)<0.
\end{equation}
\end{theorem}
\end{theorybox}

% \huan{Need to clear distinction of auto\_LiRPA and alpha-beta-CROWN. What is the difference? Need to mention that auto\_LiRPA implements which part of Section III, alpha-beta-CROWN implements which part of Section III.}
\section{Usage of \abCROWN}\label{sec:abcrown_usage}
The \abCROWN solver provides a unified, high-level entry point for working with neural networks in control and analysis tasks. Regardless of the task, the workflow consistently follows four main steps:

\begin{enumerate}
    \item \textbf{Computation Graph:} Define the system dynamics and neural network as a PyTorch Module. To ensure compatibility with the solver, the graph must adhere to the following rules:
    \begin{itemize}
        \item \textbf{Single Tensor Output:} The module must return exactly one tensor output, whose bounds will be computed. If multiple outputs are needed, they must be concatenated into a single tensor.
        \item \textbf{Batch Dimension:} Each input must include a leading \emph{batch dimension}, which should not be modified within the function.
    \end{itemize}
    \item \textbf{Variables and Solver Initialization:} Define symbolic input and output variables (e.g., \texttt{input\_vars}, \texttt{output\_vars}) that match the dimensions of your computation graph, and instantiate the \texttt{ABCrownSolver}.
    \item \textbf{Constraints:} Formulate the input domain and property targets using the \texttt{IOConstraints} Domain-Specific Language (DSL).
    \item \textbf{Execution:} Call the appropriate function for your mode (\texttt{verify()}, \texttt{compute\_bounds()}, or \texttt{minimize()}/\texttt{maximize()}) passing the defined constraints.
\end{enumerate}

\noindent
To see how these steps come together in practice, it is helpful to first understand the three primary ways to use the solver:

\begin{itemize}
    \item \textbf{Verification:} This mode mathematically proves whether a specific logical property (such as safety or stability) strictly holds across an entire continuous set of inputs.
    \item \textbf{Bound Computation:} Instead of answering a yes/no verification question, this mode calculates the guaranteed numerical lower and upper bounds of the computation graph's outputs for a given input domain. This is particularly useful for finding reachable sets.
    \item \textbf{Optimization:} This mode runs constrained optimization directly over the input domain to find a specific input that minimizes or maximizes an objective. It returns the optimal objective value (primal value) and the corresponding input that achieves it. The solver simply uses CROWN's bounding capabilities internally to efficiently guide the search for this optimal solution.
\end{itemize}

\noindent
We will illustrate the usage with the following examples.

\subsection{Verification Mode: Lyapunov Stability}

In verification mode, the solver attempts to mathematically prove that a specific property holds across an entire continuous input domain. For a system like the Van der Pol oscillator, we can use this mode to verify its neural Lyapunov stability \cite{li2025two}. 

The condition for stability is that if the state $x$ is within a specific level set shell $c_1 \leq V(x) \leq c_2$, then its time derivative $\dot{V}(x)$ must be strictly negative. The implication condition can be converted into an equivalent disjunction as in condition~\eqref{eq:continuous_stability_1}. With the mathematical formulation clear, our first practical step is to implement the computation graph that represents this system.

% We can express this logical implication as:
% \begin{equation}
%     (V(x) \geq c_1 \land V(x) \leq c_2) \to \dot{V}(x) < 0
% \end{equation}
% To feed this into the verifier, we convert the implication into an equivalent disjunction (OR logic):
% \begin{equation}
%     V(x) < c_1 \lor V(x) > c_2 \lor \dot{V}(x) < 0
% \end{equation}

\subsubsection{Computation Graph}

The graph takes the state, computes $V(x)$, calculates the control action, and simulates the dynamics. Note that \texttt{dynamics}, \texttt{controller}, and \texttt{lyapunov} passed into this graph are all expected to be pre-defined PyTorch Modules. To evaluate the derivative $\dot{V}(x)$ for this specific system, we need to compute the Jacobian $\nabla V(x)$. Instead of using standard auto-differentiation, the graph imports a specialized \texttt{JacobianOP} from \texttt{auto\_LiRPA}. When sent to the bound computation function internally, \texttt{auto\_LiRPA} automatically extends the computation graph with explicit ``gradient nodes" that encode the backward gradient propagation, allowing the solver to rigorously bound the derivative.

\begin{codebox}{Python}
import torch
import torch.nn as nn
from auto_LiRPA.jacobian import JacobianOP

class LyapunovComputationGraph(nn.Module):
    def __init__(self, dynamics, controller, lyapunov):
        super().__init__()
        # dynamics, controller, and lyapunov are all PyTorch Modules
        self.dynamics = dynamics
        self.controller = controller
        self.lyapunov = lyapunov

    def forward(self, x):
        x = x.clone().requires_grad_(True)
        V_x = self.lyapunov(x)
        u = self.controller(x)
        x_dot = self.dynamics(x, u)
        
        # Define the Jacobian operator; auto_LiRPA expands this internally
        dVdx = JacobianOP.apply(V_x, x).squeeze(1)
        V_dot = torch.sum(dVdx * x_dot, dim=1, keepdim=True)
        
        # Rule 1: Return a single concatenated tensor
        return torch.cat((V_x, V_dot), dim=1)

computation_graph = LyapunovComputationGraph(dynamics, controller, lyapunov)
\end{codebox}

\subsubsection{Variables and Solver Initialization}

We declare variables matching the dimensions of the state input and the stacked output conditions, then tie them to the computation graph by initializing the solver.

\begin{codebox}{Python}
from abcrown import (
    ABCrownSolver, ConfigBuilder, IOConstraints, input_vars, output_vars
)
# Note: These imports are needed for all the following examples.

# 2. Define graph input and output variables
x = input_vars(2)  # State vector
y = output_vars(2) # y[0] = V(x), y[1] = V_dot(x)

# 3. Create solver instance with the computation graph
cfg = ConfigBuilder.from_defaults()

solver = ABCrownSolver(computation_graph, x, y, config=cfg)
\end{codebox}

\subsubsection{Constraints}

Now that the variables are ready, we need to formalize the input limits and the property we want to verify using \texttt{IOConstraints}. Note that the API strictly requires output constraints to use \texttt{<} or \texttt{>}, not \texttt{<=} or \texttt{>=}.

\begin{codebox}{Python}
# 4. Create verification constraints
input_constraints = (x >= [-4.8, -10.8]) & (x <= [4.8, 10.8])

# Translate: V(x) < c1 OR V(x) > c2 OR V_dot(x) < 0
c1, c2 = 0.1, 2.0
output_constraints = (y[0] < c1) | (y[0] > c2) | (y[1] < 0.0)

constraints = IOConstraints(
    input_vars=x, output_vars=y,
    input_constraints=input_constraints, 
    output_constraints=output_constraints
)
\end{codebox}

\subsubsection{Execution}

Finally, we run the verification by passing the constraints to the solver.

\begin{codebox}{Python}
# 5. Run verification
result = solver.verify(constraints=constraints) 

print(f"Status: {result.status}, Verified: {result.success}")
\end{codebox}

\subsection{Bound Computation Mode: One-Step Reachability}
Instead of a Boolean yes/no verification, we can use the solver to directly extract guaranteed bounds on a network's outputs under a bounded input domain. This is highly applicable for tasks like one-step reachability analysis, where we want to bound the next state of a system $x_{t+1}$ given that the current state $x_t$ lies within a known hyperrectangle $\mathcal{B}_t$. 

Given discrete-time dynamics $x_{t+1} = f(x_t)$, if $x_t \in [x_{\min}, x_{\max}]$, bound computation finds tight bounds $x_{t+1, \min}$ and $x_{t+1, \max}$ such that:
\begin{equation}
    x_{t+1} \in [x_{t+1, \min}, x_{t+1, \max}]
\end{equation}

\subsubsection{Computation Graph}

Just like in verification, we start by defining the system. Here, we use a simple graph that wraps the neural network transition dynamics.

\begin{codebox}{Python}
class ReachabilityGraph(nn.Module):
    def __init__(self, nn_dynamics):
        super().__init__()
        self.nn_dynamics = nn_dynamics
        
    def forward(self, x_t):
        # x_t shape: [batch, 2]. Output represents x_{t+1}
        return self.nn_dynamics(x_t)

computation_graph = ReachabilityGraph(nn_dynamics).eval()
\end{codebox}

\subsubsection{Variables and Solver Initialization}

We define the symbolic variables corresponding to the graph's inputs and outputs, and use them to initialize the solver. By default, the solver iteratively refines its computed bounds using branch-and-bound until it reaches a standard timeout of 360 seconds. Often, however, users do not require extremely tight bounds and prefer to avoid long waiting times. To balance execution time with bound tightness, we can easily customize the configuration by limiting the execution time or the number of refinement steps. In this example, we reduce the timeout to 30 seconds and cap the maximum number of iterations at 100. It is also worth noting that setting ``\texttt{bab/max\_iterations}'' to 1 will bypass the branch-and-bound refinement entirely, quickly yielding the naive CROWN bound.

\begin{codebox}{Python}
# 2. Define graph input and output variables
x = input_vars(2)
y = output_vars(2)

# 3. Create solver instance with the computation graph
cfg = (
    ConfigBuilder.from_defaults()
    .set("bab/timeout", 30)
    .set("bab/max_iterations", 100)
)
solver = ABCrownSolver(computation_graph, x, y, config=cfg)
\end{codebox}

\subsubsection{Constraints}

For bounding, we only need to define the input domain. If no output constraint is provided, the solver defaults to computing the concrete upper and lower bounds of the objective.

\begin{codebox}{Python}
# 4. Create input constraints
# Define the set of possible current states
input_constraints = (x >= [-1.0, -1.0]) & (x <= [1.0, 1.0])

constraints = IOConstraints(
    input_vars=x, 
    input_constraints=input_constraints
)
\end{codebox}

\subsubsection{Execution}

Next, we invoke the solver by calling \texttt{.compute\_bounds()}. By default, the solver will only calculate and return the guaranteed concrete numerical bounds for the outputs.

\begin{codebox}{Python}
# 5. Compute bounds on y given constraints on x
# Default behavior: Returns a BoundsResult containing only lower and upper
bounds_result = solver.compute_bounds(
    constraints=constraints,
    objective=y
)

print("Next state lower bounds:", bounds_result.lower)
print("Next state upper bounds:", bounds_result.upper)
\end{codebox}

Additionally, \texttt{compute\_bounds()} can optionally provide the linear relaxation coefficients for the computation graph, which is incredibly useful for linearizing network dynamics in some downstream control tasks. 

By passing the argument \texttt{return\_linear\_bounds} as \texttt{True}, the returned \texttt{BoundsResult} object will also include a \texttt{linear\_bounds} variable. This returned object contains the matrices $A$ and vectors $b$ that define the linear bounds:
\begin{equation}
    A_l x + b_l \leq f(x) \leq A_u x + b_u
\end{equation}
By extracting these terms, users can explicitly reconstruct the linear relaxations.

\begin{codebox}{Python}
# Pass return_linear_bounds=True to extract relaxation terms
bounds_result_with_linear = solver.compute_bounds(
    constraints=constraints,
    objective=y,
    return_linear_bounds=True
)

if bounds_result_with_linear.linear_bounds is not None:
    print("Linear relaxation successfully computed!")
\end{codebox}

\subsection{Optimization Mode: Model Predictive Control (MPC)}
The API allows for constrained optimization, allowing you to maximize or minimize a linear objective expression directly over the input domain while strictly enforcing constraints on both the inputs and the network's outputs.

As an example, in a single-step Model Predictive Control (MPC) problem, we aim to find the optimal control action $u^*$ that minimizes a specific cost while respecting actuator limits and safety boundaries on the resulting state.

Suppose our system outputs the next state $x_{t+1} \in \mathbb{R}^2$. We want to minimize a linear combination of the next state's variables:
\begin{equation}
    \min_{u} \ x_{t+1, 1} + 0.5 x_{t+1, 2}
\end{equation}
subject to the system dynamics, the input actuator limits, and a safety constraint ensuring the first dimension of the next state does not exceed a threshold:
\begin{equation}
\begin{aligned}
    x_{t+1} &= f(x_t, u)\\
    -5 &\leq u \leq 5\\
    x_{t+1, 1} &< 2.0
\end{aligned}
\end{equation}

\subsubsection{Computation Graph}

Once again, we build the graph to represent the forward pass. The graph takes the control input $u$ and returns the simulated next state $x_{t+1}$. The initial state $x_t$ is treated as a constant inside the graph.

\begin{codebox}{Python}
class MPCGraph(nn.Module):
    def __init__(self, dynamics, x_current):
        super().__init__()
        self.dynamics = dynamics
        self.x_current = x_current 
        
    def forward(self, u):
        # Simulate next state
        x_next = self.dynamics(self.x_current, u)
        
        # Return next state (shape [batch, 2])
        return x_next

computation_graph = MPCGraph(dynamics, x_t).eval()
\end{codebox}

\subsubsection{Variables and Solver Initialization}

We declare $u$ as an input variable and bind it to the solver along with the output $y$.

\begin{codebox}{Python}
# 2. Define graph input and output variables
u = input_vars(1)
y = output_vars(2) # y represents the next state x_{t+1}

# 3. Build the solver
solver = ABCrownSolver(computation_graph, u, y)
\end{codebox}

\subsubsection{Constraints}

With the system dynamics defined, we set both the input limits and the output safety constraints. 

\begin{codebox}{Python}
# 4. Create constraints
# Actuator limits constraint (Input)
input_constraints = (u >= -5.0) & (u <= 5.0)

# Safety constraint on the next state (Output)
output_constraints = (y[0] < 2.0)

constraints = IOConstraints(
    input_vars=u, output_vars=y,
    input_constraints=input_constraints,
    output_constraints=output_constraints
)
\end{codebox}

\subsubsection{Execution}

To execute the optimization, we call \texttt{.minimize()}. We can pass a composite linear expression of the outputs to serve as our objective.

\begin{codebox}{Python}
# 5. Minimize the objective
# Define the objective as a linear expression of the outputs
objective_expr = y[0] + 0.5 * y[1]

# Returns OptimizationResult
result = solver.minimize(
    objective=objective_expr,
    constraints=constraints
)

if result.success:
    print(f"Optimal objective value (Primal Value): {result.primal_value}")
    print(f"Optimal control action (u*): {result.x_best}")
\end{codebox}

Note that the solver requires the objective expression to be a linear combination of the defined variables. If the objective is nonlinear (e.g., a quadratic cost $x_{t+1}^\top Q x_{t+1}$, the solution is simple: merge the formula to compute that nonlinear objective directly into the computation graph so that it is returned as an additional output dimension. You can then specify that single output dimension as the one to minimize or constrain.

\subsection{Satisfiability Mode: dReal-Like API}

In addition to the standard verification and optimization modes using explicit PyTorch computation graphs, $\alpha,\!\beta$-CROWN also provides a drop-in compatibility layer for the \texttt{dReal} Python API. While all SMT problems can be solved using \abCROWN's native API, users who are already familiar with using traditional SMT tools can easily use this \texttt{dReal}-like interface.

Users with existing \texttt{dReal} scripts can utilize $\alpha,\!\beta$-CROWN's solver engine without rewriting their constraint definitions or learning the new API. By simply changing the import statement, the solver seamlessly processes the standard \texttt{Variable}, \texttt{And}, and \texttt{CheckSatisfiability} commands. Here is an example demonstrating this exact 1-to-1 compatibility:
\begin{codebox}{Python}
# Simply replace 'from dreal import *' with the abcrown compatibility layer
from abcrown.abcrown_smt import *

# 1. Define variables exactly as in dReal
x = Variable("x")
y = Variable("y")
z = Variable("z")

# 2. Formulate constraints using dReal's syntax and math functions
f_sat = And(
    0 <= x, x <= 10,
    0 <= y, y <= 10,
    0 <= z, z <= 10,
    sin(x) + cos(y) == z
)

# 3. Check satisfiability with a specified delta tolerance
result = CheckSatisfiability(f_sat, 0.001)

print(result)
\end{codebox}

% \subsubsection{General Configuration File}

% \subsection{Replacing SMT solver}
% \textcolor{blue}{Find some new things, reachability maybe? Lipshitz constant.}

\section{Controller and Certificate Synthesis}
In this section, we discuss how to jointly synthesize the controller and the neural certificate through learning (e.g., a learned Lyapunov function or contraction metrics). We note that since formal verification requires the specification to hold for the entire domain, it is much more difficult compared to a normal machine learning task, where certain errors on the test set are tolerable. Therefore, in this setting, a simple random sampling approach to construct the train/test set will not suffice, and we typically design our training algorithms to be verification-aware. One typical training framework is called Counterexample-Guided Inductive Synthesis (CEGIS), which essentially involves finding counterexamples to the verification condition and using them as training data. However, CEGIS exhibits some key disadvantages. For example, training with CEGIS typically requires specially designed initializations, and the final model obtained may not be verification-friendly~\cite{li2025two,shi2024certified}.

% \textcolor{blue}{Cite more. Notations.}
Although there are heuristics that can make CEGIS work better, training a verifier-friendly model purely based on this approach remains a challenge. To conquer this, one prominent direction is certified training, where one directly incorporates the bounds generated by the verifier into the learning process~\cite{shi2024certified, zhang2019towards, shi2021fast, lee2021towards, mao2023understanding}. Compared to dReal-style solvers, a key practical advantage of $\alpha,\!\beta$-CROWN is that its linear relaxation bounds are \textbf{differentiable} with respect to network parameters. At a high level, bound propagation only depends on standard tensor operations (e.g., matrix multiplications), so the resulting output bound is differentiable. This property can be very useful in practice. For example, for a specification of the form $g(x;\theta) >  0$ over an input set $\mathcal{B}$, $\alpha,\!\beta$-CROWN produces a certified lower bound
\begin{align*}
g(x;\theta) \ge \underline{A}(\theta)x+\underline{b}(\theta),\quad \forall x\in\mathcal{B},
\end{align*}
and thus a differentiable lower bound on the worst-case violation:
\begin{align}
    \mathcal{L}(\theta) := \underline{g}(\mathcal{B};\theta) := \inf_{x\in \mathcal{B}} \ \underline{A}(\theta)x + \underline{b}(\theta).
\end{align}
Because $\mathcal{L}(\theta)$ is differentiable, it can be maximized directly using gradient-based optimization to encourage satisfaction of the original specification directly. When this objective is used during training, the procedure is often referred to as \textbf{certified training}. Since this training objective is closely related to the verification objective, models trained with this technique can be much verification-friendlier. In the control setting where the domain $\mathcal{B}$ is huge, just as we have discussed about $\alpha,\!\beta$-CROWN, we can also utilize the idea of branch and bound to tighten the bound in training time. Formally, during this training, we maintain a dataset consisting of domains that partition the full domain $\mathcal{B}$:
\begin{align}
    \mathbb{B}:= \{[\underline{x}^{(1)}, \overline{x}^{(1)}], [\underline{x}^{(2)}, \overline{x}^{(2)}], \cdots, [\underline{x}^{(n)}, \overline{x}^{(n)}]\}.
\end{align}
On each of the subdomains, we use $\alpha,\!\beta$-CROWN to obtain a lower bound $\underline{g}(\underline{x}^{(i)}, \overline{x}^{(i)})$. The loss function can then be defined as
\begin{align}
    L(\theta; \mathbb{B}) = \mathbb{E}_{\mathbb{B}} [\mathrm{ReLU}(-\underline{g}(\underline{x}^{(i)}, \overline{x}^{(i)}))].
\end{align}
Moreover, we can still incorporate the counterexample loss used in the CEGIS framework as an added empirical guard. We shall note that since now the domains are split into smaller pieces, the counterexample finding procedure can also more efficiently find counterexamples and help the learning. When the bound is not tight enough so that the minimization fails to push the lower bound above zero in a particular training subdomain, we further branch it with the heuristics described in section~\ref{sec:bab}. This branching can further tighten the bounds. We then add the split domains into the training data and remove the original domain. 

We note that since this approach directly targets the verification condition and explicitly uses the bound generated by the verifier as a training objective, it can yield models that are much easier to verify. As a comparison, in the 2D quadrotor environment, the certified training approach can yield a model that is 5 times faster to verify compared to a CEGIS-trained model. 

\section{Advanced Topics}

\subsection{Solving Optimization Problems}\label{sec:optimization}
In this section, we briefly discuss how bound propagation, together with branch and bound, can be used as an efficient nonlinear programming solver. In this case, rather than verifying a specific condition like $F(x) \ge 0$ for all $x \in \mathcal{B}$, our goal is to find the minimizer of $F$ over this box domain, i.e., we wish to solve
\begin{align}
    f^* = \min_{x\in \mathcal{B}} F(x).
\end{align}
To efficiently compute this, we shall still continue with the certified lower bound computation with branch-and-bound. We notice that on each of the subdomains $\mathcal{B}_i$, in addition to the certified lower bounds $\underline{F}_{i}$ being produced, a PGD run as in the first stage of solving the satisfiability problem could be used to produce feasible upper bounds $\overline{F}_{i}$ of the problem on each subdomain. Using this upper bound information, it is clear that we can always prune subdomains $\mathcal{B}_j$ with $\underline{F}_{j} > \min_i \overline{F}_{i}$, since this subdomain is certified to be above the best solution we have found so far. Now, with continued branch-and-bound, most subdomains will be pruned, leaving the found feasible solution closer to the optimal possible solution. The algorithm will terminate when a computation budget is reached. We note that in this optimization context, when the gap between the best feasible upper bound and the minimum certified lower bound over the remaining subdomains is small, this branch-and-bound procedure also provides an optimality certificate.

\subsection{Other Improvements of Verification}
When bounds produced by standard propagation remain loose, the verifier supports several advanced functionalities for further bound tightening. Specifically, $\alpha$-CROWN~\cite{xu2020automatic} tightens bounds by optimizing parameterized local linear relaxations. To resolve persistent looseness from non-linearities, $\beta$-CROWN~\cite{wang2021beta} and GenBaB~\cite{shi2024genbab} perform activation branching by partitioning the domains of non-linear operators and optimizing the resulting Lagrangian dual variables, which can achieve complete verification without relying on input domain branch-and-bound. Furthermore, Clip-and-Verify~\cite{zhou2025clip} exploits otherwise discarded affine constraints to continuously prune the input domain and tighten intermediate bounds. Finally, extensions such as GCP-CROWN~\cite{zhang2022general} and BICCOS~\cite{zhou2024scalable} further strengthen the relaxations by incorporating cutting planes directly into the bound propagation framework.

\section{CONCLUSIONS}
In this tutorial, we presented a unified workflow for certifying learned controllers with the neural network verifier $\alpha,\!\beta$-CROWN. The same framework supports multiple control tasks and usage modes. We also discussed the API of $\alpha,\!\beta$-CROWN, as well as training paradigms such as CEGIS and certified training for co-synthesizing controllers and verifier-friendly certificates. We hope this tutorial helps bridge control theory and scalable neural network verification, and makes formal certification a more practical component of learning-based control design.

\section{ACKNOWLEDGMENTS}

 B. Hu is generously supported by the NSF award CAREER-2048168 and  the AFOSR award FA9550-23-1-0732.
 H. Zhang, X. Zhong, and H. Li are supported in part by the AI2050 program at Schmidt Sciences (AI2050 Early Career Fellowship) and NSF (IIS-2331967).

% \huan{Add reference}

% A conclusion section is not required. Although a conclusion may review the main points of the paper, do not replicate the abstract as the conclusion. A conclusion might elaborate on the importance of the work or suggest applications and extensions. 


%%%%%%%%%%%%%%%%%%%%%%%%%%%%%%%%%%%%%%%%%%%%%%%%%%%%%%%%%%%%%%%%%%%%%%%%%%%%%%%%

%%%%%%%%%%%%%%%%%%%%%%%%%%%%%%%%%%%%%%%%%%%%%%%%%%%%%%%%%%%%%%%%%%%%%%%%%%%%%%%%



% \section*{ACKNOWLEDGMENT}

% The preferred spelling of the word ÒacknowledgmentÓ in America is without an ÒeÓ after the ÒgÓ. Avoid the stilted expression, ÒOne of us (R. B. G.) thanks . . .Ó  Instead, try ÒR. B. G. thanksÓ. Put sponsor acknowledgments in the unnumbered footnote on the first page.

% %%%%%%%%%%%%%%%%%%%%%%%%%%%%%%%%%%%%%%%%%%%%%%%%%%%%%%%%%%%%%%%%%%%%%%%%%%%%%%%%

% References are important to the reader; therefore, each citation must be complete and correct. If at all possible, references should be commonly available publications.
% \newpage
\bibliographystyle{IEEEtran}
\bibliography{reference}

@article{bucsoniu2018reinforcement,
  title={Reinforcement learning for control: Performance, stability, and deep approximators},
  author={Bu{\c{s}}oniu, Lucian and De Bruin, Tim and Toli{\'c}, Domagoj and Kober, Jens and Palunko, Ivana},
  journal={Annual Reviews in Control},
  volume={46},
  pages={8--28},
  year={2018},
  publisher={Elsevier}
}

@article{kaufmann2023champion,
  title={Champion-level drone racing using deep reinforcement learning},
  author={Kaufmann, Elia and Bauersfeld, Leonard and Loquercio, Antonio and M{\"u}ller, Matthias and Koltun, Vladlen and Scaramuzza, Davide},
  journal={Nature},
  volume={620},
  number={7976},
  pages={982--987},
  year={2023},
  publisher={Nature Publishing Group UK London}
}

@article{glavic2019deep,
  title={(Deep) reinforcement learning for electric power system control and related problems: A short review and perspectives},
  author={Glavic, Mevludin},
  journal={Annual Reviews in Control},
  volume={48},
  pages={22--35},
  year={2019},
  publisher={Elsevier}
}

@article{chang2019neural,
  title={Neural lyapunov control},
  author={Chang, Ya-Chien and Roohi, Nima and Gao, Sicun},
  journal={Advances in neural information processing systems},
  volume={32},
  year={2019}
}

@article{dai2021lyapunov,
  title={Lyapunov-stable neural-network control},
  author={Dai, Hongkai and Landry, Benoit and Yang, Lujie and Pavone, Marco and Tedrake, Russ},
  journal={arXiv preprint arXiv:2109.14152},
  year={2021}
}

@article{wu2023neural,
  title={Neural lyapunov control for discrete-time systems},
  author={Wu, Junlin and Clark, Andrew and Kantaros, Yiannis and Vorobeychik, Yevgeniy},
  journal={Advances in neural information processing systems},
  volume={36},
  pages={2939--2955},
  year={2023}
}

@article{yang2024lyapunov,
  title={Lyapunov-stable neural control for state and output feedback: A novel formulation},
  author={Yang, Lujie and Dai, Hongkai and Shi, Zhouxing and Hsieh, Cho-Jui and Tedrake, Russ and Zhang, Huan},
  journal={arXiv preprint arXiv:2404.07956},
  year={2024}
}

@article{wang2024actor,
  title={Actor-critic physics-informed neural lyapunov control},
  author={Wang, Jiarui and Fazlyab, Mahyar},
  journal={IEEE Control Systems Letters},
  year={2024},
  publisher={IEEE}
}

@article{lyapunov1992general,
  title={The general problem of the stability of motion},
  author={Lyapunov, Aleksandr Mikhailovich},
  journal={International journal of control},
  volume={55},
  number={3},
  pages={531--534},
  year={1992},
  publisher={Taylor \& Francis}
}

@article{liu2025physics,
  title={Physics-informed neural network Lyapunov functions: PDE characterization, learning, and verification},
  author={Liu, Jun and Meng, Yiming and Fitzsimmons, Maxwell and Zhou, Ruikun},
  journal={Automatica},
  volume={175},
  pages={112193},
  year={2025},
  publisher={Elsevier}
}

@article{li2025neural,
  title={Neural Contraction Metrics with Formal Guarantees for Discrete-Time Nonlinear Dynamical Systems},
  author={Li, Haoyu and Zhong, Xiangru and Hu, Bin and Zhang, Huan},
  journal={arXiv preprint arXiv:2504.17102},
  year={2025}
}

@article{madry2017towards,
  title={Towards deep learning models resistant to adversarial attacks},
  author={Madry, Aleksander and Makelov, Aleksandar and Schmidt, Ludwig and Tsipras, Dimitris and Vladu, Adrian},
  journal={arXiv preprint arXiv:1706.06083},
  year={2017}
}

@article{tedrake2010lqr,
  title={LQR-trees: Feedback motion planning via sums-of-squares verification},
  author={Tedrake, Russ and Manchester, Ian R and Tobenkin, Mark and Roberts, John W},
  journal={The International Journal of Robotics Research},
  volume={29},
  number={8},
  pages={1038--1052},
  year={2010},
  publisher={SAGE Publications Sage UK: London, England}
}

@inproceedings{dai2023convex,
  title={Convex synthesis and verification of control-Lyapunov and barrier functions with input constraints},
  author={Dai, Hongkai and Permenter, Frank},
  booktitle={2023 American Control Conference (ACC)},
  pages={4116--4123},
  year={2023},
  organization={IEEE}
}

@inproceedings{robey2020learning,
  title={Learning control barrier functions from expert demonstrations},
  author={Robey, Alexander and Hu, Haimin and Lindemann, Lars and Zhang, Hanwen and Dimarogonas, Dimos V and Tu, Stephen and Matni, Nikolai},
  booktitle={2020 59th IEEE Conference on Decision and Control (CDC)},
  pages={3717--3724},
  year={2020},
  organization={Ieee}
}

@inproceedings{dawson2022safe,
  title={Safe nonlinear control using robust neural lyapunov-barrier functions},
  author={Dawson, Charles and Qin, Zengyi and Gao, Sicun and Fan, Chuchu},
  booktitle={Conference on Robot Learning},
  pages={1724--1735},
  year={2022},
  organization={PMLR}
}

@inproceedings{gao2013dreal,
  title={dReal: An SMT solver for nonlinear theories over the reals},
  author={Gao, Sicun and Kong, Soonho and Clarke, Edmund M},
  booktitle={International conference on automated deduction},
  pages={208--214},
  year={2013},
  organization={Springer}
}

@article{zhang2018efficient,
  title={Efficient Neural Network Robustness Certification with General Activation Functions},
  author={Zhang, Huan and Weng, Tsui-Wei and Chen, Pin-Yu and Hsieh, Cho-Jui and Daniel, Luca},
  journal={Advances in Neural Information Processing Systems},
  volume={31},
  pages={4939--4948},
  year={2018},
  url={https://arxiv.org/pdf/1811.00866.pdf}
}

@article{xu2020automatic,
  title={Automatic perturbation analysis for scalable certified robustness and beyond},
  author={Xu, Kaidi and Shi, Zhouxing and Zhang, Huan and Wang, Yihan and Chang, Kai-Wei and Huang, Minlie and Kailkhura, Bhavya and Lin, Xue and Hsieh, Cho-Jui},
  journal={Advances in Neural Information Processing Systems},
  volume={33},
  year={2020}
}

@article{salman2019convex,
  title={A Convex Relaxation Barrier to Tight Robustness Verification of Neural Networks},
  author={Salman, Hadi and Yang, Greg and Zhang, Huan and Hsieh, Cho-Jui and Zhang, Pengchuan},
  journal={Advances in Neural Information Processing Systems},
  volume={32},
  pages={9835--9846},
  year={2019}
}

@inproceedings{xu2021fast,
    title={{Fast and Complete}: Enabling Complete Neural Network Verification with Rapid and Massively Parallel Incomplete Verifiers},
    author={Kaidi Xu and Huan Zhang and Shiqi Wang and Yihan Wang and Suman Jana and Xue Lin and Cho-Jui Hsieh},
    booktitle={International Conference on Learning Representations},
    year={2021},
    url={https://openreview.net/forum?id=nVZtXBI6LNn}
}

@article{wang2021beta,
  title={{Beta-CROWN}: Efficient bound propagation with per-neuron split constraints for complete and incomplete neural network verification},
  author={Wang, Shiqi and Zhang, Huan and Xu, Kaidi and Lin, Xue and Jana, Suman and Hsieh, Cho-Jui and Kolter, J Zico},
  journal={Advances in Neural Information Processing Systems},
  volume={34},
  year={2021}
}

@article{zhang2022general,
  title={General Cutting Planes for Bound-Propagation-Based Neural Network Verification},
  author={Zhang, Huan and Wang, Shiqi and Xu, Kaidi and Li, Linyi and Li, Bo and Jana, Suman and Hsieh, Cho-Jui and Kolter, J Zico},
  journal={Advances in Neural Information Processing Systems},
  year={2022}
}

@inproceedings{kotha2023provably,
 author = {Kotha, Suhas and Brix, Christopher and Kolter, J. Zico and Dvijotham, Krishnamurthy and Zhang, Huan},
 booktitle = {Advances in Neural Information Processing Systems},
 editor = {A. Oh and T. Neumann and A. Globerson and K. Saenko and M. Hardt and S. Levine},
 pages = {80270--80290},
 publisher = {Curran Associates, Inc.},
 title = {Provably Bounding Neural Network Preimages},
 volume = {36},
 year = {2023}
}

@inproceedings{zhou2024scalable,
  title={Scalable Neural Network Verification with Branch-and-bound Inferred Cutting Planes},
  author={Zhou, Duo and Brix, Christopher and Hanasusanto, Grani A and Zhang, Huan},
  booktitle={The Thirty-eighth Annual Conference on Neural Information Processing Systems},
  year={2024}
}

@inproceedings{shi2024genbab,
  title={Neural Network Verification with Branch-and-Bound for General Nonlinearities},
  author={Shi, Zhouxing and Jin, Qirui and Kolter, Zico and Jana, Suman and Hsieh, Cho-Jui and Zhang, Huan},
  booktitle={International Conference on Tools and Algorithms for the Construction and Analysis of Systems},
  year={2025}
}

@article{ames2016control,
  title={Control barrier function based quadratic programs for safety critical systems},
  author={Ames, Aaron D and Xu, Xiangru and Grizzle, Jessy W and Tabuada, Paulo},
  journal={IEEE Transactions on Automatic Control},
  volume={62},
  number={8},
  pages={3861--3876},
  year={2016},
  publisher={IEEE}
}

@article{dawson2023safe,
  title={Safe control with learned certificates: A survey of neural lyapunov, barrier, and contraction methods for robotics and control},
  author={Dawson, Charles and Gao, Sicun and Fan, Chuchu},
  journal={IEEE Transactions on Robotics},
  volume={39},
  number={3},
  pages={1749--1767},
  year={2023},
  publisher={IEEE}
}

@article{tsukamoto2020neural,
  title={Neural contraction metrics for robust estimation and control: A convex optimization approach},
  author={Tsukamoto, Hiroyasu and Chung, Soon-Jo},
  journal={IEEE Control Systems Letters},
  volume={5},
  number={1},
  pages={211--216},
  year={2020},
  publisher={IEEE}
}

@inproceedings{sun2021learning,
  title={Learning certified control using contraction metric},
  author={Sun, Dawei and Jha, Susmit and Fan, Chuchu},
  booktitle={Conference on Robot Learning},
  pages={1519--1539},
  year={2021},
  organization={PMLR}
}

@book{khalil2002nonlinear,
  title={Nonlinear systems},
  author={Khalil, Hassan K and Grizzle, Jessy W},
  volume={3},
  year={2002},
  publisher={Prentice hall Upper Saddle River, NJ}
}

@article{abate2020formal,
  title={Formal synthesis of Lyapunov neural networks},
  author={Abate, Alessandro and Ahmed, Daniele and Giacobbe, Mirco and Peruffo, Andrea},
  journal={IEEE Control Systems Letters},
  volume={5},
  number={3},
  pages={773--778},
  year={2020},
  publisher={IEEE}
}

@article{yin2021stability,
  title={Stability analysis using quadratic constraints for systems with neural network controllers},
  author={Yin, He and Seiler, Peter and Arcak, Murat},
  journal={IEEE Transactions on Automatic Control},
  volume={67},
  number={4},
  pages={1980--1987},
  year={2021},
  publisher={IEEE}
}

@inproceedings{abate2018counterexample,
  title={Counterexample guided inductive synthesis modulo theories},
  author={Abate, Alessandro and David, Cristina and Kesseli, Pascal and Kroening, Daniel and Polgreen, Elizabeth},
  booktitle={International Conference on Computer Aided Verification},
  pages={270--288},
  year={2018},
  organization={Springer}
}

@inproceedings{meng2024zubov,
  title={Zubov-Koopman Learning of Maximal Lyapunov Functions},
  author={Meng, Yiming and Zhou, Ruikun and Liu, Jun},
  booktitle={2024 American Control Conference (ACC)},
  pages={4020--4025},
  year={2024},
  organization={IEEE}
}

@article{zhao2024applications,
  title={Applications of machine learning in real-time control systems: a review},
  author={Zhao, Xiaoning and Sun, Yougang and Li, Yanmin and Jia, Ning and Xu, Junqi},
  journal={Measurement Science and Technology},
  year={2024},
  publisher={IOP Publishing}
}

@article{tsukamoto2020stochastic,
  title={Neural stochastic contraction metrics for learning-based control and estimation},
  author={Tsukamoto, Hiroyasu and Chung, Soon-Jo and Slotine, Jean-Jacques E},
  journal={IEEE Control Systems Letters},
  volume={5},
  number={5},
  pages={1825--1830},
  year={2020},
  publisher={IEEE}
}

@inproceedings{tsukamoto2021theoretical,
  title={A theoretical overview of neural contraction metrics for learning-based control with guaranteed stability},
  author={Tsukamoto, Hiroyasu and Chung, Soon-Jo and Slotine, Jean-Jacques and Fan, Chuchu},
  booktitle={2021 60th IEEE Conference on Decision and Control (CDC)},
  pages={2949--2954},
  year={2021},
  organization={IEEE}
}

@article{serry2025safe,
  title={Safe Domains of Attraction for Discrete-Time Nonlinear Systems: Characterization and Verifiable Neural Network Estimation},
  author={Serry, Mohamed and Li, Haoyu and Zhou, Ruikun and Zhang, Huan and Liu, Jun},
  journal={arXiv preprint arXiv:2506.13961},
  year={2025}
}

@article{meng2025learning,
  title={Learning regions of attraction in unknown dynamical systems via Zubov-Koopman lifting: Regularities and convergence},
  author={Meng, Yiming and Zhou, Ruikun and Liu, Jun},
  journal={IEEE Transactions on Automatic Control},
  year={2025},
  publisher={IEEE}
}

@inproceedings{liu2024tool,
  title={TOOL LyZNet: A lightweight python tool for learning and verifying neural lyapunov functions and regions of attraction},
  author={Liu, Jun and Meng, Yiming and Fitzsimmons, Maxwell and Zhou, Ruikun},
  booktitle={Proceedings of the 27th ACM International Conference on Hybrid Systems: Computation and Control},
  pages={1--8},
  year={2024}
}

@inproceedings{edwards2024fossil,
  title={Fossil 2.0: Formal certificate synthesis for the verification and control of dynamical models},
  author={Edwards, Alec and Peruffo, Andrea and Abate, Alessandro},
  booktitle={Proceedings of the 27th ACM International Conference on Hybrid Systems: Computation and Control},
  pages={1--10},
  year={2024}
}

@article{jagtap2020formal,
  title={Formal synthesis of stochastic systems via control barrier certificates},
  author={Jagtap, Pushpak and Soudjani, Sadegh and Zamani, Majid},
  journal={IEEE Transactions on Automatic Control},
  volume={66},
  number={7},
  pages={3097--3110},
  year={2020},
  publisher={IEEE}
}

@inproceedings{de2008z3,
  title={Z3: An efficient SMT solver},
  author={De Moura, Leonardo and Bj{\o}rner, Nikolaj},
  booktitle={International conference on Tools and Algorithms for the Construction and Analysis of Systems},
  pages={337--340},
  year={2008},
  organization={Springer}
}

@article{li2025two,
  title={Two-Stage Learning of Stabilizing Neural Controllers via Zubov Sampling and Iterative Domain Expansion},
  author={Li, Haoyu and Zhong, Xiangru and Hu, Bin and Zhang, Huan},
  journal={arXiv preprint arXiv:2506.01356},
  year={2025}
}

@article{gowal2018effectiveness,
  title={On the effectiveness of interval bound propagation for training verifiably robust models},
  author={Gowal, Sven and Dvijotham, Krishnamurthy and Stanforth, Robert and Bunel, Rudy and Qin, Chongli and Uesato, Jonathan and Arandjelovic, Relja and Mann, Timothy and Kohli, Pushmeet},
  journal={arXiv preprint arXiv:1810.12715},
  year={2018}
}

@article{zhou2025clip,
  title={Clip-and-verify: Linear constraint-driven domain clipping for accelerating neural network verification},
  author={Zhou, Duo and Chavez, Jorge and Chen, Hesun and Hanasusanto, Grani A and Zhang, Huan},
  journal={arXiv preprint arXiv:2512.11087},
  year={2025}
}

@article{liu2025formal,
  title={Formal Verification of Control Lyapunov-Barrier Functions for Safe Stabilization with Bounded Controls},
  author={Liu, Jun},
  journal={arXiv preprint arXiv:2511.10510},
  year={2025}
}

@article{meng2025towards,
  title={Towards Learning and Verifying Maximal Lyapunov-Barrier Functions with a Zubov PDE Formulation},
  author={Meng, Yiming and Liu, Jun},
  journal={arXiv preprint arXiv:2511.09523},
  year={2025}
}

@inproceedings{liu2025formally,
  title={Formally verified physics-informed neural control lyapunov functions},
  author={Liu, Jun and Fitzsimmons, Maxwell and Zhou, Ruikun and Meng, Yiming},
  booktitle={2025 American Control Conference (ACC)},
  pages={1347--1354},
  year={2025},
  organization={IEEE}
}

@inproceedings{mirman2018differentiable,
  title={Differentiable abstract interpretation for provably robust neural networks},
  author={Mirman, Matthew and Gehr, Timon and Vechev, Martin},
  booktitle={International Conference on Machine Learning},
  pages={3578--3586},
  year={2018},
  organization={PMLR}
}

@article{bunel2018unified,
  title={A unified view of piecewise linear neural network verification},
  author={Bunel, Rudy R and Turkaslan, Ilker and Torr, Philip and Kohli, Pushmeet and Mudigonda, Pawan K},
  journal={Advances in neural information processing systems},
  volume={31},
  year={2018}
}

@article{bof2018lyapunov,
  title={Lyapunov theory for discrete time systems},
  author={Bof, Nicoletta and Carli, Ruggero and Schenato, Luca},
  journal={arXiv preprint arXiv:1809.05289},
  year={2018}
}

@article{tjeng2017evaluating,
  title={Evaluating robustness of neural networks with mixed integer programming},
  author={Tjeng, Vincent and Xiao, Kai and Tedrake, Russ},
  journal={arXiv preprint arXiv:1711.07356},
  year={2017}
}

@inproceedings{katz2017reluplex,
  title={Reluplex: An efficient SMT solver for verifying deep neural networks},
  author={Katz, Guy and Barrett, Clark and Dill, David L and Julian, Kyle and Kochenderfer, Mykel J},
  booktitle={International conference on computer aided verification},
  pages={97--117},
  year={2017},
  organization={Springer}
}

@inproceedings{katz2019marabou,
  title={The marabou framework for verification and analysis of deep neural networks},
  author={Katz, Guy and Huang, Derek A and Ibeling, Duligur and Julian, Kyle and Lazarus, Christopher and Lim, Rachel and Shah, Parth and Thakoor, Shantanu and Wu, Haoze and Zelji{\'c}, Aleksandar and others},
  booktitle={International conference on computer aided verification},
  pages={443--452},
  year={2019},
  organization={Springer}
}

@article{shi2024certified,
  title={Certified Training with Branch-and-Bound for Lyapunov-stable Neural Control},
  author={Shi, Zhouxing and Li, Haoyu and Hsieh, Cho-Jui and Zhang, Huan},
  journal={arXiv preprint arXiv:2411.18235},
  year={2024}
}

@inproceedings{wong2018provable,
  title={Provable defenses against adversarial examples via the convex outer adversarial polytope},
  author={Wong, Eric and Kolter, Zico},
  booktitle={International conference on machine learning},
  pages={5286--5295},
  year={2018},
  organization={PMLR}
}

@article{raghunathan2018certified,
  title={Certified defenses against adversarial examples},
  author={Raghunathan, Aditi and Steinhardt, Jacob and Liang, Percy},
  journal={arXiv preprint arXiv:1801.09344},
  year={2018}
}

@inproceedings{brown2022unified,
  title={A unified view of SDP-based neural network verification through completely positive programming},
  author={Brown, Robin A and Schmerling, Edward and Azizan, Navid and Pavone, Marco},
  booktitle={International conference on artificial intelligence and statistics},
  pages={9334--9355},
  year={2022},
  organization={PMLR}
}

@inproceedings{gehr2018ai2,
  title={Ai2: Safety and robustness certification of neural networks with abstract interpretation},
  author={Gehr, Timon and Mirman, Matthew and Drachsler-Cohen, Dana and Tsankov, Petar and Chaudhuri, Swarat and Vechev, Martin},
  booktitle={2018 IEEE symposium on security and privacy (SP)},
  pages={3--18},
  year={2018},
  organization={IEEE}
}

@article{singh2018fast,
  title={Fast and effective robustness certification},
  author={Singh, Gagandeep and Gehr, Timon and Mirman, Matthew and P{\"u}schel, Markus and Vechev, Martin},
  journal={Advances in neural information processing systems},
  volume={31},
  year={2018}
}

@inproceedings{chen2013flow,
  title={Flow*: An analyzer for non-linear hybrid systems},
  author={Chen, Xin and {\'A}brah{\'a}m, Erika and Sankaranarayanan, Sriram},
  booktitle={International Conference on Computer Aided Verification},
  pages={258--263},
  year={2013},
  organization={Springer}
}

@article{althoff2021set,
  title={Set propagation techniques for reachability analysis},
  author={Althoff, Matthias and Frehse, Goran and Girard, Antoine},
  journal={Annual Review of Control, Robotics, and Autonomous Systems},
  volume={4},
  number={1},
  pages={369--395},
  year={2021},
  publisher={Annual Reviews}
}

@inproceedings{tran2020nnv,
  title={NNV: the neural network verification tool for deep neural networks and learning-enabled cyber-physical systems},
  author={Tran, Hoang-Dung and Yang, Xiaodong and Manzanas Lopez, Diego and Musau, Patrick and Nguyen, Luan Viet and Xiang, Weiming and Bak, Stanley and Johnson, Taylor T},
  booktitle={International conference on computer aided verification},
  pages={3--17},
  year={2020},
  organization={Springer}
}

@article{shen2024bab,
  title={Bab-nd: Long-horizon motion planning with branch-and-bound and neural dynamics},
  author={Shen, Keyi and Yu, Jiangwei and Barreiros, Jose and Zhang, Huan and Li, Yunzhu},
  journal={arXiv preprint arXiv:2412.09584},
  year={2024}
}

@article{tran2018convergence,
  title={Convergence properties for discrete-time nonlinear systems},
  author={Tran, Duc N and R{\"u}ffer, Bj{\"o}rn S and Kellett, Christopher M},
  journal={IEEE Transactions on Automatic Control},
  volume={64},
  number={8},
  pages={3415--3422},
  year={2018},
  publisher={IEEE}
}

@inproceedings{shen2026diffreach,
  title={Parallel Differentiable Reachability for Learning and Planning with Certified Neural Dynamics and Controllers},
  author={Keyi Shen and Glen Chou},
  booktitle={Proceedings of Robotics: Science and Systems (RSS)},
  year={2026}
}

@article{zhang2019towards,
  title={Towards stable and efficient training of verifiably robust neural networks},
  author={Zhang, Huan and Chen, Hongge and Xiao, Chaowei and Gowal, Sven and Stanforth, Robert and Li, Bo and Boning, Duane and Hsieh, Cho-Jui},
  journal={arXiv preprint arXiv:1906.06316},
  year={2019}
}

@article{shi2021fast,
  title={Fast certified robust training with short warmup},
  author={Shi, Zhouxing and Wang, Yihan and Zhang, Huan and Yi, Jinfeng and Hsieh, Cho-Jui},
  journal={Advances in Neural Information Processing Systems},
  volume={34},
  pages={18335--18349},
  year={2021}
}

@article{lee2021towards,
  title={Towards better understanding of training certifiably robust models against adversarial examples},
  author={Lee, Sungyoon and Lee, Woojin and Park, Jinseong and Lee, Jaewook},
  journal={Advances in Neural Information Processing Systems},
  volume={34},
  pages={953--964},
  year={2021}
}

@article{mao2023understanding,
  title={Understanding certified training with interval bound propagation},
  author={Mao, Yuhao and M{\"u}ller, Mark Niklas and Fischer, Marc and Vechev, Martin},
  journal={arXiv preprint arXiv:2306.10426},
  year={2023}
}

@article{tran2020verification,
  title={Verification approaches for learning-enabled autonomous cyber--physical systems},
  author={Tran, Hoang-Dung and Xiang, Weiming and Johnson, Taylor T},
  journal={IEEE Design \& Test},
  volume={39},
  number={1},
  pages={24--34},
  year={2020},
  publisher={IEEE}
}

@inproceedings{liu2023towards,
  title={Towards learning and verifying maximal neural Lyapunov functions},
  author={Liu, Jun and Meng, Yiming and Fitzsimmons, Maxwell and Zhou, Ruikun},
  booktitle={2023 62nd IEEE Conference on Decision and Control (CDC)},
  pages={8012--8019},
  year={2023},
  organization={IEEE}
}

@inproceedings{dai2024verification,
  title={Verification and synthesis of compatible control lyapunov and control barrier functions},
  author={Dai, Hongkai and Jiang, Chuanrui and Zhang, Hongchao and Clark, Andrew},
  booktitle={2024 IEEE 63rd Conference on Decision and Control (CDC)},
  pages={8178--8185},
  year={2024},
  organization={IEEE}
}

@article{yin2021imitation,
  title={Imitation learning with stability and safety guarantees},
  author={Yin, He and Seiler, Peter and Jin, Ming and Arcak, Murat},
  journal={IEEE Control Systems Letters},
  volume={6},
  pages={409--414},
  year={2021},
  publisher={IEEE}
}

@inproceedings{papachristodoulou2005tutorial,
  title={A tutorial on sum of squares techniques for systems analysis},
  author={Papachristodoulou, Antonis and Prajna, Stephen},
  booktitle={Proceedings of the 2005, American Control Conference, 2005.},
  pages={2686--2700},
  year={2005},
  organization={IEEE}
}

@inproceedings{newton2022stability,
  title={Stability of non-linear neural feedback loops using sum of squares},
  author={Newton, Matthew and Papachristodoulou, Antonis},
  booktitle={2022 IEEE 61st Conference on Decision and Control (CDC)},
  pages={6000--6005},
  year={2022},
  organization={IEEE}
}

@article{detailleur2025improved,
  title={Improved Sum-of-Squares Stability Verification of Neural-Network-Based Controllers},
  author={Detailleur, Alvaro and Ducard, Guillaume and Onder, Christopher},
  journal={arXiv preprint arXiv:2507.10352},
  year={2025}
}

@article{clark2024semialgebraic,
  title={A semialgebraic framework for verification and synthesis of control barrier functions},
  author={Clark, Andrew},
  journal={IEEE Transactions on Automatic Control},
  volume={70},
  number={5},
  pages={3101--3116},
  year={2024},
  publisher={IEEE}
}

@article{wei2021control,
  title={Control contraction metric synthesis for discrete-time nonlinear systems},
  author={Wei, Lai and Mccloy, Ryan and Bao, Jie},
  journal={IFAC-PapersOnLine},
  volume={54},
  number={3},
  pages={661--666},
  year={2021},
  publisher={Elsevier}
}

@inproceedings{karg2020stability,
  title={Stability and feasibility of neural network-based controllers via output range analysis},
  author={Karg, Benjamin and Lucia, Sergio},
  booktitle={2020 59th IEEE Conference on Decision and Control (CDC)},
  pages={4947--4954},
  year={2020},
  organization={IEEE}
}

@InProceedings{pmlr-v270-hu25a,
  title = 	 {Verification of Neural Control Barrier Functions with Symbolic Derivative Bounds Propagation},
  author =       {Hu, Hanjiang and Yang, Yujie and Wei, Tianhao and Liu, Changliu},
  booktitle = 	 {Proceedings of The 8th Conference on Robot Learning},
  pages = 	 {1797--1814},
  year = 	 {2025},
  volume = 	 {270},
  series = 	 {Proceedings of Machine Learning Research},
  month = 	 {06--09 Nov},
  publisher =    {PMLR},
}

@article{zhang2023exact,
  title={Exact verification of relu neural control barrier functions},
  author={Zhang, Hongchao and Wu, Junlin and Vorobeychik, Yevgeniy and Clark, Andrew},
  journal={Advances in neural information processing systems},
  volume={36},
  pages={5685--5705},
  year={2023}
}

@article{vertovec2025scalable,
  title={Scalable Verification of Neural Control Barrier Functions Using Linear Bound Propagation},
  author={Vertovec, Nikolaus and Mathiesen, Frederik Baymler and Badings, Thom and Laurenti, Luca and Abate, Alessandro},
  journal={arXiv preprint arXiv:2511.06341},
  year={2025}
}

@inproceedings{hu2024real,
  index = {C75},
  title = {Real-Time Safe Control of Neural Network Dynamic Models with Sound Approximation},
  author = {Hu, Hanjiang and Lan, Jianglin and Liu, Changliu},
  booktitle = {Learning for Dynamics and Control Conference},
  year = {2024},
  repository = {https://github.com/intelligent-control-lab/BOND},
  url = {https://proceedings.mlr.press/v242/hu24a.html}
}

@article{li2025verifiable,
  index = {J37},
  title = {Verifiable Safety Q-Filters via Hamilton-Jacobi Reachability and Multiplicative Q-Networks},
  author = {Li, Jiaxing and Hu, Hanjiang and Yang, Yujie and Liu, Changliu},
  journal = {IEEE Control Systems Letters},
  year = {2025},
  url = {https://ieeexplore.ieee.org/stamp/stamp.jsp?tp=&arnumber=11157757}
}

@inproceedings{zhao2023safety,
  index = {C56},
  title = {Safety index synthesis via sum-of-squares programming},
  author = {Zhao, Weiye and He, Tairan and Wei, Tianhao and Liu, Simin and Liu, Changliu},
  booktitle = {American Control Conference},
  pages = {732--737},
  year = {2023},
  organization = {IEEE}
}

@article{jafarpour2024efficient,
  title={Efficient interaction-aware interval analysis of neural network feedback loops},
  author={Jafarpour, Saber and Harapanahalli, Akash and Coogan, Samuel},
  journal={IEEE Transactions on Automatic Control},
  volume={69},
  number={12},
  pages={8706--8721},
  year={2024},
  publisher={IEEE}
}

@inproceedings{jafarpour2023interval,
  title={Interval reachability of nonlinear dynamical systems with neural network controllers},
  author={Jafarpour, Saber and Harapanahalli, Akash and Coogan, Samuel},
  booktitle={Learning for Dynamics and Control Conference},
  pages={12--25},
  year={2023},
  organization={PMLR}
}

@article{zhang2024reachability,
  title={Reachability analysis of neural network control systems with tunable accuracy and efficiency},
  author={Zhang, Yuhao and Zhang, Hang and Xu, Xiangru},
  journal={IEEE Control Systems Letters},
  volume={8},
  pages={1697--1702},
  year={2024},
  publisher={IEEE}
}

@inproceedings{zhang2023reachability,
  title={Reachability analysis of neural network control systems},
  author={Zhang, Chi and Ruan, Wenjie and Xu, Peipei},
  booktitle={Proceedings of the AAAI Conference on Artificial Intelligence},
  volume={37},
  number={12},
  pages={15287--15295},
  year={2023}
}

@article{abate2024safe,
  title={Safe reach set computation via neural barrier certificates},
  author={Abate, Alessandro and Bogomolov, Sergiy and Edwards, Alec and Potomkin, Kostiantyn and Soudjani, Sadegh and Zuliani, Paolo},
  journal={IFAC-PapersOnLine},
  volume={58},
  number={11},
  pages={107--114},
  year={2024},
  publisher={Elsevier}
}

@article{ravanbakhsh2015counterexample,
  title={Counterexample guided synthesis of switched controllers for reach-while-stay properties},
  author={Ravanbakhsh, Hadi and Sankaranarayanan, Sriram},
  journal={arXiv preprint arXiv:1505.01180},
  year={2015}
}

@article{masti2023counter,
  title={Counter-example guided inductive synthesis of control Lyapunov functions for uncertain systems},
  author={Masti, Daniele and Fabiani, Filippo and Gnecco, Giorgio and Bemporad, Alberto},
  journal={IEEE Control Systems Letters},
  volume={7},
  pages={2047--2052},
  year={2023},
  publisher={IEEE}
}

@inproceedings{dai2020counter,
  title={Counter-example guided synthesis of neural network Lyapunov functions for piecewise linear systems},
  author={Dai, Hongkai and Landry, Benoit and Pavone, Marco and Tedrake, Russ},
  booktitle={2020 59th IEEE Conference on Decision and Control (CDC)},
  pages={1274--1281},
  year={2020},
  organization={IEEE}
}

@inproceedings{ravanbakhsh2016robust,
  title={Robust controller synthesis of switched systems using counterexample guided framework},
  author={Ravanbakhsh, Hadi and Sankaranarayanan, Sriram},
  booktitle={Proceedings of the 13th International Conference on Embedded Software},
  pages={1--10},
  year={2016}
}

@article{coogan2017finite,
  title={Finite abstraction of mixed monotone systems with discrete and continuous inputs},
  author={Coogan, Samuel and Arcak, Murat},
  journal={Nonlinear Analysis: Hybrid Systems},
  volume={23},
  pages={254--271},
  year={2017},
  publisher={Elsevier}
}

@inproceedings{coogan2015efficient,
  title={Efficient finite abstraction of mixed monotone systems},
  author={Coogan, Samuel and Arcak, Murat},
  booktitle={Proceedings of the 18th International Conference on Hybrid Systems: Computation and Control},
  pages={58--67},
  year={2015}
}

@inproceedings{coogan2020mixed,
  title={Mixed monotonicity for reachability and safety in dynamical systems},
  author={Coogan, Samuel},
  booktitle={2020 59th IEEE Conference on Decision and Control (CDC)},
  pages={5074--5085},
  year={2020},
  organization={IEEE}
}

@article{dutreix2020specification,
  title={Specification-guided verification and abstraction refinement of mixed monotone stochastic systems},
  author={Dutreix, Maxence and Coogan, Samuel},
  journal={IEEE Transactions on Automatic Control},
  volume={66},
  number={7},
  pages={2975--2990},
  year={2020},
  publisher={IEEE}
}

@article{sidrane2022overt,
  title={Overt: An algorithm for safety verification of neural network control policies for nonlinear systems},
  author={Sidrane, Chelsea and Maleki, Amir and Irfan, Ahmed and Kochenderfer, Mykel J},
  journal={Journal of Machine Learning Research},
  volume={23},
  number={117},
  pages={1--45},
  year={2022}
}

\end{document}